\documentclass[conference]{IEEEtran}
\IEEEoverridecommandlockouts

\usepackage{pdfsync}
\usepackage{refcount}
\usepackage{cite}
\ifCLASSINFOpdf
   \usepackage[pdftex]{graphicx}
\usepackage[cmex10]{amsmath}
\usepackage{amssymb,amsthm,alltt}
\usepackage{thmtools,thm-restate}
\usepackage[all]{xy}
\xyoption{poly}
\usepackage{stackrel}
\usepackage{xspace}
 \usepackage[usenames]{color}
\usepackage{logic-article}
\usepackage{stmaryrd} 
\usepackage{marginnote}
\usepackage{hyperref}	
\usepackage[textsize=small]{todonotes}
\usepackage{array}
\usepackage{tikz}
\usetikzlibrary{patterns}
\usepackage{subcaption}
\usepackage{paralist}

\newcounter{quotecount}

\begin{document}

\title{Timed pushdown automata revisited}

\author{\IEEEauthorblockN{Lorenzo Clemente\thanks{The first author acknowledges a partial support of the Polish National Science Centre grant 2013/09/B/ST6/01575.}}
\IEEEauthorblockA{University of Warsaw}
\and
\IEEEauthorblockN{S{\l}awomir Lasota\thanks{The last author acknowledges a partial support of the Polish National Science Centre grant 2012/07/B/ST6/01497.}}
\IEEEauthorblockA{University of Warsaw}
}

\maketitle

\begin{abstract}
	This paper contains two results on timed extensions of pushdown automata (PDA).
	As our first result we prove that the model of dense-timed PDA
	of Abdulla \emph{et al.}
	collapses: it is expressively equivalent to dense-timed PDA with \emph{timeless} stack.
	Motivated by this result, we advocate the framework of first-order definable PDA,
	a specialization of PDA in sets with atoms, as the right setting to define and investigate timed extensions of PDA.
	The general model obtained in this way is Turing complete.
	As our second result we prove NEXPTIME upper complexity bound for the non-emptiness problem for an expressive subclass.
	As a byproduct, we obtain a tight EXPTIME complexity bound for a more restrictive subclass of PDA with timeless stack,
	thus subsuming the complexity bound known for dense-timed PDA.
\end{abstract}

\IEEEpeerreviewmaketitle


\newcommand{\mysubsubsection}[1]{\vspace{1mm} \noindent{\bf #1. }}
\newcommand{\introsubsection}[1]{\vspace{1mm} \noindent{\bf #1. }}

\newcommand{\lorenzo}[1]{\todo[color=blue!20]{#1}}
\newcommand{\nop}{\mathsf{nop}}
\newcommand{\push}[1]{\mathsf{push}(#1)}
\newcommand{\pop}[1]{\mathsf{pop}(#1)}
\newcommand{\pushpop}{\text{\sc push-pop}}
\newcommand{\pushpopA}{\pushpop^0}
\newcommand{\pushpopB}{\pushpop^1}
\newcommand{\pushpopC}{\pushpop^2}
\newcommand{\sep}{\ | \ }
\newcommand{\true}{\mathbf{t}}
\newcommand{\safe}[1]{\mathsf{safe}({#1})}
\newcommand{\eventually}[1]{\mathsf{eventually}({#1})}
\newcommand{\test}[1]{\mathsf{test}({#1})}
\newcommand{\dtPDA}{dtPDA\xspace}

\newcommand{\g}{{\cal G}}
\newcommand{\wlg}{w.l.o.g.\xspace}
\newcommand{\Wlog}{W.l.o.g.\xspace}

\newcommand{\timed}{timed\xspace}

\newcommand{\tr}{timed register\xspace}
\newcommand{\Tr}{Timed register\xspace}
\newcommand{\trPDA}{{trPDA}\xspace}
\newcommand{\TrPDA}{{TrPDA}\xspace}
\newcommand{\trCFG}{{trCFG}\xspace}
\newcommand{\TrCFG}{{TrCFG}\xspace}

\newcommand{\atoms}{\mathbb A}
\newcommand{\nat}{\mathbb N}
\newcommand{\N}{\mathbb N}
\newcommand{\Q}{\mathbb Q}
\newcommand{\R}{\mathbb R}
\newcommand{\D}{{\mathbb D}}
\newcommand{\Z}{{\mathbb Z}}

\newcommand{\Aa}{{\cal A}}
\newcommand{\Bb}{{\cal B}}
\newcommand{\Cc}{{\cal C}}
\newcommand{\Qq}{{\hat Q}}
\newcommand{\Oo}{{\cal O}}
\newcommand{\Hh}{{\cal H}}
\newcommand{\aaut}{{\cal A}}
\newcommand{\baut}{{\cal B}}
\newcommand{\caut}{{\cal C}}
\newcommand{\taut}{{\cal T}}
\newcommand{\uaut}{{\cal U}}

\newcommand{\psQ}{\ddot Q}
\newcommand{\psorb}{\text{orbits}(\psQ)}
\newcommand{\dist}[2]{\text{d}(#1, #2)}
\newcommand{\psQf}{\tilde Q}
\newcommand{\psorbf}{\text{orbits}(\psQf)}
\newcommand{\pspsQ}{\dddot Q}
\newcommand{\pspspsQ}{\ddddot Q}

\newcommand{\reg}{{\cal R}}
\newcommand{\regrules}{{\cal S}}
\newcommand{\const}{c}
\newcommand{\kmax}{k_{\text{max}}}

\newcommand{\trans}[3]{ #1 \stackrel{#2}{\longrightarrow} #3 }
\newcommand{\trrule}[3]{(#1,#3)}
\newcommand{\transtrans}[3]{ #1 \stackrel{#2}{\longrightarrow}^* #3 }
\newcommand{\horiztrans}[3]{ #1 \stackrel{#2}{\leadsto} #3 }
\newcommand{\horiztranstrans}[3]{ #1 \stackrel{#2}{\leadsto}^* #3 }
\newcommand{\horiztransplus}[3]{ #1 \stackrel{#2}{\leadsto}^+ #3 }
\newcommand{\invreach}[2]{\text{Reach}^{-1}_{#1}(#2)}
\newcommand{\invacc}[1]{\text{Reach}^{-1}_{#1}}
\newcommand{\ew}{\varepsilon}
\newcommand{\eps}{\tau} 
\newcommand{\slcomm}[1]{\marginpar{\vspace{1mm}\textcolor{blue}{\tiny SL: #1}}}
\newcommand{\exptime}{{\tt ExpTime}\xspace}
\newcommand{\exptimec}{{\tt ExpTime}-complete\xspace}
\newcommand{\nexptime}{{\tt NExpTime}\xspace}

\newcommand{\eqdef}{\stackrel {\text{def}} =}
\renewcommand{\setminus}{-}
\newcommand{\func}[3]{\mathop{\mathchoice{%
    {#1}\colon\;{#2}\;\longrightarrow\;{#3}}{%
    {#1}\colon{#2}\to{#3}}{%
    script}{%
    sscript}
  }}
\newcommand{\defin}[1]{[#1]}
\newcommand{\ov}[1]{\overline{#1}}
\newcommand{\aut}[1]{\text{Aut}(#1)}
\newcommand{\orbit}[1]{\text{orbit}(#1)}
\newcommand{\fodef}{FO-definable\xspace}
\newcommand{\Fodef}{FO-definable\xspace}
\newcommand{\qfdef}{qf-definable\xspace}
\newcommand{\Qfdef}{Qf-definable\xspace}
\newcommand{\prestar}[1]{\text{Pre}^*(#1)}
\newcommand{\autom}{{\cal A}}
\newcommand{\quot}[2]{#1/{#2}}
\newcommand{\set}[1]{\left\{ #1 \right\}}
\newcommand{\setof}[2]{\set{#1 \; | \; #2}}
\newcommand{\tuple}[1]{(#1)}  
\newcommand{\goesto}[1]{\stackrel {#1} \longrightarrow}
\newcommand{\lang}[1]{\mathcal L({#1})}
\newcommand{\poly}{\text{poly}}

\newcommand{\bigrule}[2]{\frac{#1}{#2}}
\renewcommand{\phi}{\varphi}
\newcommand{\vars}[1]{\text{vars}(#1)}

\newcommand{\rhopush}{\text{\sc push}}
\newcommand{\rhopop}{\text{\sc pop}}
\newcommand{\rhonop}{\text{\sc nop}}
\newcommand{\concat}{\cdot}
\newcommand{\elt}[1]{\langle #1\rangle}

\newcommand{\inhabited}[1]{\vdash{#1}}

\newcommand{\len}[1]{|#1|}

\newcommand{\ignore}[1]{}


\section{Introduction}

\introsubsection{Background}
Timed automata~\cite{AD94} are a popular model of time-dependent behavior.
A timed automaton is a finite automaton extended with a finite number of variables,
called clocks, that can be reset and tested for inequalities with integers; so equipped, a timed automaton
can read timed words, whose letters are labeled with real (or rational) timestamps. 
The value of a clock implicitly increases with the elapse of time, which is modeled 
by monotonically increasing time-stamps of input letters. 

In this paper, we investigate timed automata extended with a stack.
An early model extending timed automata with an untimed stack,
which we call \emph{pushdown timed automata} (PDTA),
has been considered by Bouajjani \emph{et al.}~\cite{bouajjani:timed:PDA:94}.
Intuitively, PDTA recognize timed languages that can be obtained
by extending an untimed context-free language with regular timing constraints.
A more expressive model, called \emph{recursive timed automata} (RTA),
has been independently proposed (in an essentially equivalent form)
by Trivedi and Wojtczak \cite{TrivediWojtczak:RTA:2010},
and by Benerecetti \emph{et al.} \cite{BenerecettiMinopoliPeron:RTA:2010}.
RTA use a timed stack to store the current clock valuation, which can be restored at the time of pop.
This facility makes RTA able to recognize timed language with non-regular timing constraints (unlike PDTA).

More recently, \emph{dense-timed pushdown automata} (\dtPDA)
have been proposed by Abdulla \emph{et al.}~\cite{AbdullaAtigStenman:DensePDA:12}
as yet another extension of PDTA.
In \dtPDA, a clock may be pushed on the stack,
and its value increases with the elapse of time, exactly like the value of an ordinary clock.
When popped from the stack, the value may be tested for inequalities with integers.
The non-emptiness problem for \dtPDA is solved in \cite{AbdullaAtigStenman:DensePDA:12}
by an ingenious reduction to non-emptiness of classical \emph{untimed} PDA.
As a byproduct, this shows that the untiming projection of \dtPDA-language is context-free.
Perhaps surprisingly, we prove the semantic collapse of \dtPDA to PDTA,
i.e., \dtPDA with timeless stack but timed control locations:
every \dtPDA may be effectively transformed into a PDTA \emph{that recognizes the same timed language}.
Notice that this is much stronger than a mere reduction of the non-emptiness problem from the former to the latter model.
Intuitively, the collapse is caused by the accidental interference of the LIFO stack discipline with the monotonicity of time,
combined with the restrictions on stack operations assumed in \dtPDA.
Thus, \dtPDA are equivalent to PDTA, and therefore included in RTA.
The collapse motivates the quest for a more expressive framework for timed extensions of PDA.

\introsubsection{Timed register pushdown automata}
We advocate \emph{sets with atoms} as the right setting for defining and investigating timed extensions of various classes of automata.
This setting is parametrized by a logical structure $\atoms$, called \emph{atoms}.
Intuitively speaking, sets with atoms are very much like classical sets,
but the notion of finiteness is relaxed to orbit-finiteness,
i.e., finiteness up to an automorphism of atoms $\atoms$.
The relaxation of finiteness allows to capture naturally various infinite-state models.
For instance, ignoring some inessential details, register automata~\cite{FK94} (recognizing data languages)
are expressively equivalent to the reinterpretation of the classical definition of `finite NFA' as `orbit-finite NFA' in sets with 
\emph{equality atoms} $(\N, =)$ (see~\cite{BKL11full} for details),
and analogously for register pushdown automata~\cite{ChengKaminski:CFL:AI98}. 

Along similar lines, timed automata (without stack) are essentially a subclass of NFA in sets with 
\emph{timed atoms} \mbox{$(\Q, \leq, +1)$}, i.e.,
rationals with the natural order and the $+1$ function
(see~\cite{BL12icalp} for details).
The automorphisms of timed atoms are thus monotonic bijections
from $\Q$ to $\Q$ that preserve integer differences.
In fact, to capture timed automata it is enough to work in a well-behaved subclass of
sets with timed atoms, namely in \emph{(first-order) definable sets}. Examples of definable 
sets are 
\begin{align*}
A \ & = \ \{ (x, y, z) \in \Q^3 \ : \ x < y < z+1 < x+4 \}\\
A' \ & = \ \{ (x, y) \in \Q^2 \ : \ x = y \ \lor \ y > x+2\}.
\end{align*}
The first one is orbit-finite, while the other is not.

By reinterpreting the classical definition of PDA in definable sets we obtain a powerful model,
which we call \emph{timed register PDA} (\trPDA), where, roughly speaking, 
a clock (or even a tuple of clocks) may be pushed to, and popped from the stack, conditioned by arbitrary
clock constraints referring possibly to other clocks.
Notice that monotonicity is not part of the definition of timed atoms,
and thus in general \trPDA read non-monotonic timed words,
unlike classical timed automata or dense-timed PDA.
This is not a restriction, since monotonicity can be checked by the automaton itself,
and thus we can model monotonic as well as non-monotonic timed languages.
An example language recognized by a \trPDA (or even by \trCFG)
is the language of palindromes over the alphabet $A$ defined above.
Another example is the language of bracket expressions over the alphabet $\{ [,]\} \times \Q$,
where the timestamps of every pair of matching brackets belong to $A'$.
These languages intuitively require a timed stack in order to be recognized,
and thus fall outside the class of \dtPDA due to our collapse result.

\begin{figure}
	
	\centering{\includegraphics[scale=0.25]{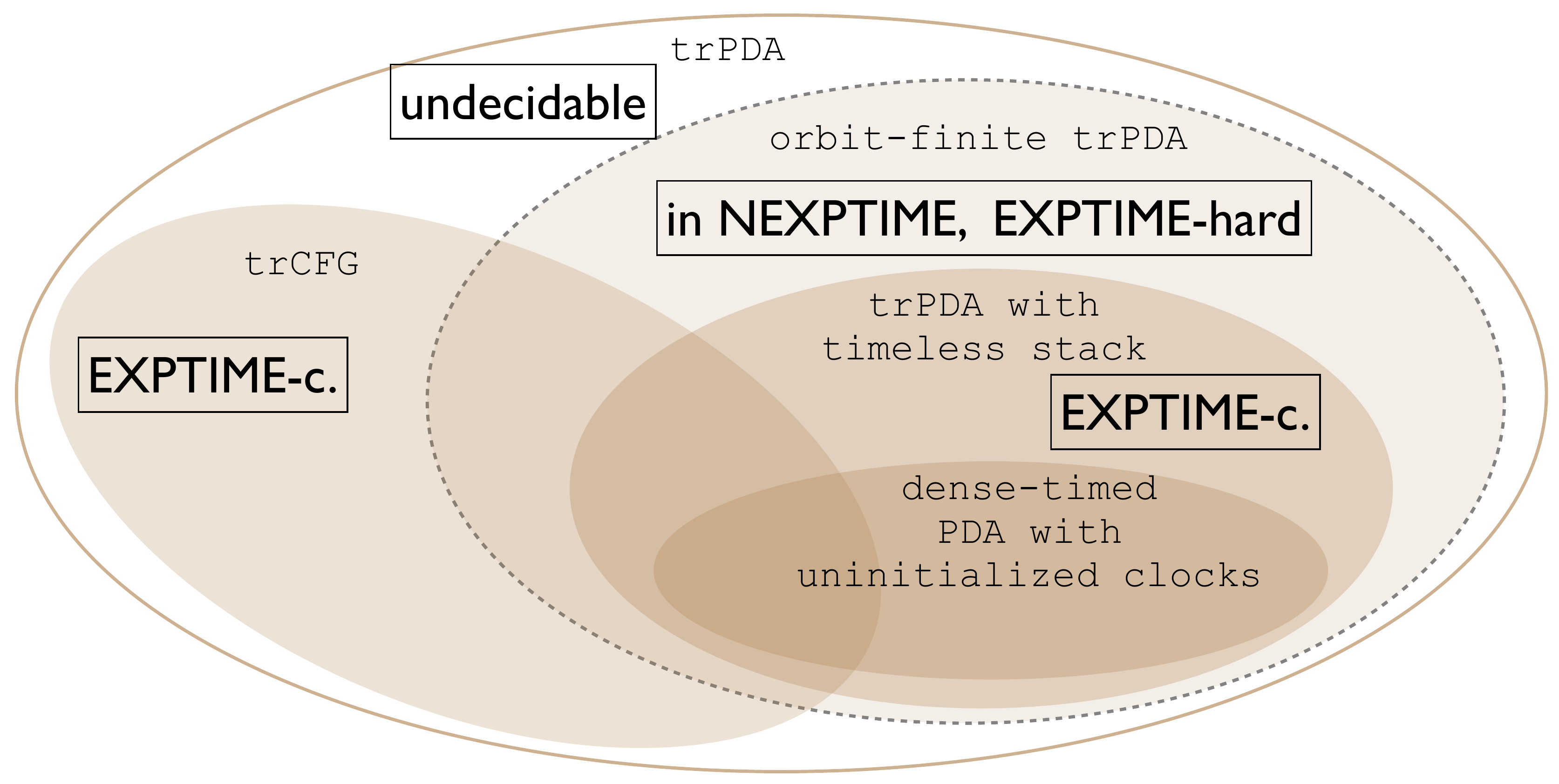}}
	\caption{Classes of timed pushdown languages.}
	\label{fig:languages}
	
\end{figure}

\introsubsection{Contributions}
In view of possible applications to verification of time-dependent recursive programs,
we focus on the computational complexity of the non-emptiness problem for \trPDA.
We isolate several interesting classes of \trPDA,
which are summarized in Fig.~\ref{fig:languages}.
All intersections are non-trivial.
Our model subsume \dtPDA,
for the simple reason that the finite-state control is essentially a timed-register NFA,
which subsumes timed automata, i.e., the finite-state control of \dtPDA.
For the general model we prove undecidability of non-emptiness. 
This motivates us to distinguish an expressive subclass,
which we call \emph{orbit-finite} \trPDA,
which is obtained from the general model by imposing a certain orbit-finiteness restriction on push and pop operations.
We show that non-emptiness of orbit-finite \trPDA is in \nexptime.
This is shown by reduction to non-emptiness of the least solution of 
a system of equations over sets of integers (cf.~\cite{JO11} and references therein).
This reduction is the technical core of the paper.
Moreover, it shows the essentially \emph{quantitative} flavor of the dense time domain $(\Q, \leq, +1)$
as opposed to other kind of atoms, like equality $(\N, =)$ or total-order atoms $(\Q, \leq)$.
Note that $(\R, \leq, +1)$ has the same first-order theory of the rationals,
and thus considering the latter instead of the reals is with no loss of generality.
Interestingly, our proofs work just as well over the discrete time domain $(\Z, \leq, +1)$.

In order to establish the claimed complexity upper bound,
we establish, along the way, tight complexity results for solving systems of equations in special form.
From this analysis, we derive \exptime-completeness of the subclass of \trPDA with timeless stack.
Due to our collapse result, under a simple technical assumption that preserves non-emptiness,
\dtPDA can be effectively transformed into \trPDA with timeless stack,
and thus we subsume the \exptime upper bound shown in~\cite{AbdullaAtigStenman:DensePDA:12}.

Finally, we consider the reinterpretation of context-free grammars in sets with timed atoms.
We prove that timed context-free languages are a strict subclass of \trPDA languages,
and that their non-emptiness is \exptime-complete.

Except for the technical results, the paper offers a wider perspective on modeling timed systems.
We claim that sets with  atoms have a significant and still unexplored potential 
for capturing timed extensions of classical models of computation.

\introsubsection{Organization}
In Sec.~\ref{sec:dtPDA} we show the collapse result for \dtPDA.
In Sec.~\ref{sec:defin} we introduce the setting of definable sets.
Then, in Sec.~\ref{sec:trPDA} we define \trPDA and its subclasses, formulate our complexity results,
and relate in detail these results to the previously known \exptime-completeness of \dtPDA.
The following Sec.~\ref{sec:proof} is the core technical part of the paper and it is devoted to the proofs of the upper bounds.
The last section contains final remarks and sketch of future work.
The missing parts of the proofs are delegated to the appendix.


\section{Dense-timed pushdown automata}  \label{sec:dtPDA}

As the first result of the paper, we show that \dtPDA as proposed by \cite{AbdullaAtigStenman:DensePDA:12}
recognize the same timed languages as its variant with timeless stack.
This result is much stronger than the reduction proposed in \cite{AbdullaAtigStenman:DensePDA:12},
which shows that \dtPDA and its variant with timeless stack are equivalent w.r.t. the \emph{untimed} language
(as opposed to the full timed language).
In fact, we even prove this for a non-trivial generalization of the model of \cite{AbdullaAtigStenman:DensePDA:12} with diagonal pop constraints (cf. below).
In view of our collapse result,
we abuse terminology and we also call the extended model \dtPDA.
A \emph{clock constraint} over a set of clocks $X$ is a formula $\phi$ generated by the following grammar:
\begin{align*}
	\phi \ ::= \	\true \sep
					x \sim k \sep
					x - y \sim k \sep
					\phi \wedge \phi,
\end{align*}
where $\true$ is the trivial constraint which is always true,
$x, y \in X$, $k \in \Z$, and $\sim \ \in \{ <, \leq, \geq, > \}$.
We do not have disjunction $\vee$ since it can be simulated by nondeterminism in the transition relation of the automaton.
We write $y - x \sim k \in \phi$ to denote that $y - x \sim k$ is a conjunct in $\phi$.
A \emph{dense-timed pushdown automaton} (\dtPDA) 
is a tuple $\taut = (L, l_0, \Sigma, \Gamma, X, z, \Delta)$
where $L$ is a finite set of control locations,
$l_0 \in L$ is the initial location,
$\Sigma$ is a finite input alphabet,
$\Gamma$ is a finite stack alphabet,
$X$ is a finite set of clocks,
and $z$ is a special clock not in $X$ representing the age of the topmost stack symbol.
The last item $\Delta$ is a set of transition rules of the form:
$l \goesto {a, \phi, Y, op} l'$
with $l, l' \in L$ control locations,
$a \in \Sigma_\varepsilon = \Sigma \cup \{\varepsilon\}$  an input letter, 
$\phi$ a constraint over clocks in $X$,
a subset $Y \subseteq X$ of clocks that will be reset,
and $op$ is either $\nop$, $\pop{\alpha \models \psi_0}$, or $\push{\alpha \models \psi_1}$,
where $\alpha \in \Gamma$ a stack symbol,
$\psi_0$ a constraint over clocks in $X \cup \{z\}$ (called \emph{pop constraint})
and $\psi_1$ a constraint over $\{z\}$ (called \emph{push constraint}).
An automaton has \emph{timeless stack}
if all its pop operations $\pop{\alpha \models \true}$ have the trivial constraint $\true$,
in which case just write $\pop{\alpha}$.
%

The formal semantics of \dtPDA follows~\cite{AbdullaAtigStenman:DensePDA:12},
and can be found in the appendix.
Intuitively, every symbol on the stack carries a nonnegative rational number representing its age.
Ages increase monotonically as time elapses, all at the same rate, and at the same rate as the other clocks of the automaton.
Every time a new symbol is pushed on the stack,
its age is nondeterministically initialized to a value of $z$ satisfying the push constraint $\psi_1$,
and it can be popped only if its current age satisfies the constraint $\psi_0$.
Note that the push constraint $\psi_1$ essentially forces the initial age into a (possibly unbounded) interval.
The original definition of~\cite{AbdullaAtigStenman:DensePDA:12} imposed the same restriction on pop constraints.
Our definition of pop constraint is more liberal,
since we allow more general \emph{diagonal pop constraints} of the form $z - x \sim k$.
Despite this seemingly more general definition,
we show nonetheless that the stack can be made \emph{timeless} while preserving the timed language recognized by the automaton.
\begin{restatable}{theorem}{thmDTPDAuntimedstack}
	\label{thm:Abdulla}
	A \dtPDA $\taut$ can be effectively transformed into a \dtPDA $\uaut$ with timeless stack 
	recognizing the same timed language.
	Moreover, $\uaut$ has linearly many clocks w.r.t. $\taut$,
	and exponentially many control locations. 
\end{restatable}
\noindent
%

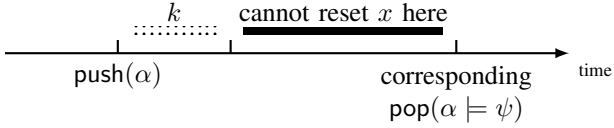
\begin{figure}
	
		\begin{tikzpicture}[x=1.5cm]
	
			\draw[->,thick,>=latex] (0,0) -- (5,0) node [below right] {$\scriptstyle \text{time}$};
			\draw[thick] (1,0) -- ++(0, 5pt) node [below = 1ex] {$\push\alpha$};
			\draw[thick] (2,0) -- ++(0, 5pt);
			\draw[thick] (4,0) -- ++(0, 5pt) node [below = 1ex] {$\begin{array}{cc}\text{corresponding} \\ \pop{\alpha \models \psi}\end{array}$};

			\fill[pattern=dots]([xshift=5pt]1,0.25) rectangle node[above] {$k$} ([xshift=-5pt]2,0.35); 
			\fill([xshift=5pt]2,0.25) rectangle node[above] {cannot reset $x$ here} ([xshift=-5pt]4,0.35);  
		
		\end{tikzpicture}
		
		
\ignore{
	\\[4ex]
	\begin{subfigure}[b]{\linewidth}
		\centering
		\begin{tikzpicture}[x=1.5cm]
	
			\draw[->,thick,>=latex] (0,0) -- (5,0) node [below right] {$\scriptstyle \text{time}$};
			\draw[thick] (0,0) -- ++(0, 5pt);
			\draw[thick] (1,0) -- ++(0, 5pt) node [below = 1ex] {$\push\alpha$};
			\draw[thick] (4,0) -- ++(0, 5pt) node [below = 1ex] {$\begin{array}{cc}\text{corresponding} \\ \pop{\alpha \models \psi}\end{array}$};

			\fill[pattern=dots]([xshift=5pt]0,0.25) rectangle node[above] {$-k$} ([xshift=-5pt]1,0.33); 
			\fill([xshift=5pt]0,0.85) rectangle node[above] {cannot reset $x$ here} ([xshift=-5pt]4,.95);  

		\end{tikzpicture}
		
		\caption{Case $k\leq0$.}
		
	\end{subfigure}
}
	
	\caption{Reset restriction on $x$ when $z - x \lesssim k \in \psi$.}
	\label{fig:reset_restriction}
	
\end{figure}


\begin{figure}
	
	\centering	
	\begin{tikzpicture}[x=1.5cm]

		\draw[->,thick,>=latex] (0,0) -- (5,0) node [below right] {$\scriptstyle \text{time}$};
		\draw[thick] (0,0) -- ++(0, 5pt) node [above = 1ex] {$\push\alpha$};
		\draw[thick] (1,0) -- ++(0, -5pt) node [below = 1ex, xshift = -2ex] {$\push\beta$};
		\draw[thick] (1.5,0) -- ++(0, 5pt);
		\draw[thick] (2.5,0) -- ++(0, -5pt);
		\draw[thick] (4,0) -- ++(0, -5pt) node [below = 1ex, xshift = 2ex] {$\pop{\beta}$};
		\draw[thick] (4.5,0) -- ++(0, 5pt) node [above = 1ex] {$\pop{\alpha}$};

		\fill[pattern=dots]([xshift=5pt]0,0.25) rectangle node[above] {$k$} ([xshift=-5pt]1.5,0.35); 
		\fill([xshift=5pt]1.5,0.25) rectangle node[above] {no reset of $x$} ([xshift=-5pt]4.5,0.35); 

		\fill[pattern=dots]([xshift=5pt]1,-.22) rectangle node[below] {$k$} ([xshift=-5pt]2.5,-.33); 
		\fill([xshift=5pt]2.5,-.25) rectangle node[below] {$\begin{array}{cc} \textrm{no reset} \\ \textrm{of } x \end{array}$} ([xshift=-5pt]4,-.35); 
	
	\end{tikzpicture}
		
	\caption{Current reset restrictions always subsume new ones.}
	\label{fig:reset_restriction_subsumption}
	
\end{figure}
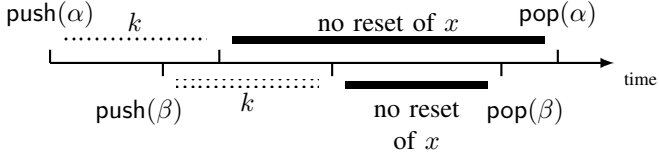

\begin{proof}[Proof (sketch)]

	We explain here the basic idea of the transformation.
	The formal construction can be found in the appendix.
	W.l.o.g. we assume that:
	\begin{inparaenum}
		\item Pop constraints are conjunctions of formulae of the form $z - x \sim k$,
		\item transition rules involving a push or pop operation never reset clocks, and
		\item the initial age of stack symbols pushed on the stack is always 0.
	\end{inparaenum}
	These assumptions will simplify the construction;
	we show in the appendix how an automaton can be modified in order to satisfy them.
	The intuition is that a pop constraint of the form $z - x \lesssim k$ with $\lesssim \ \in \{ <, \leq \}$
	implies that clock $x$ cannot be reset after $k$ (possibly negative) time units of the push and before the corresponding pop.
	We call this a \emph{reset restriction}; cf. Fig.~\ref{fig:reset_restriction}.
	We call a pop constraint $z - x \lesssim k$ \emph{active} if it has been guessed to hold at the time of a future pop.
	To keep track of reset restrictions,
	we carry in the control state a set $R$ of tuples of the form $(x, \lesssim, k)$ for every active pop constraint.
	An extra clock $\hat x_{\lesssim k}$ which is reset at time of push
	is used to check $\hat x_{\lesssim k} \lesssim k$ whenever $x$ is reset
	in order to guarantee that $x$ is not reset too late.
	If $k \leq 0$, then we need to additionally check that $x$ was not reset within the last $-k$ time units,
	which amounts to check $x \gtrsim -k$ at the time of push.
	The crucial observation is that, if a new reset restriction $(x, \lesssim, k)$ arises for an already active constraint,
	then we can safely ignore it 
	since it is always subsumed by the current one.
	In other words, whenever the old restriction is satisfied,
	so is the new one, which is thus redundant; cf. Fig.~\ref{fig:reset_restriction_subsumption}.
		
	The situation for a pop constraint of the form $z - x \gtrsim k$ with $\gtrsim \in \{ >, \geq \}$ is dual,
	since it requires that clock $x$ is reset after at least $k$ (possibly negative) time units of the push and before its corresponding pop.
	We call this a \emph{reset obligation}; cf. Fig.~\ref{fig:reset_obligation}.
	We keep track of a set $O$ of tuples $(x, \gtrsim, k)$ for every active pop constraint $z - x \gtrsim k$,
	meaning that clock $x$ must be reset before the \emph{next pop}.
	When $x$ is reset after $k$ time units of the push, we remove $(x, \gtrsim, k)$ from $O$.
	To verify the latter condition, we use an additional clock $\hat x_{\gtrsim k}$ which is reset at the time of push,
	and we check that $\hat x_{\gtrsim k} \gtrsim k$ holds.
	A new reset obligation with $k \leq 0$ is discarded if $-x \gtrsim k$ already holds at the time of push.
	A pop is allowed only if $O$ is empty, i.e., all reset obligations have been satisfied.
	The crucial observation is that
	a new reset obligation $(x, \gtrsim, k)$ always subsumes one already in $O$,
	in the sense that, whenever the former is satisfied, so it is the latter; cf. Fig.~\ref{fig:reset_obligation_subsumption}.
	Thus, previous obligations are always discarded in favor of new ones
	(this is dual w.r.t. reset restrictions).
	When there is a new push, we have to additionally guess whether obligations in $O$ not subsumed by new ones
	will be satisfied either before the matching pop, or after it.
	In the first case they are kept in $O$,
	while in the second case they are pushed on the stack in order to be put back into $O$ at the matching pop.
	%
\end{proof}
\noindent
The construction uses $\varepsilon$-transitions,
which simplifies substantially the encoding.
A more complex construction not using $\varepsilon$-transitions can be given,
and thus the collapse holds even for $\varepsilon$-free \dtPDA.
We don't know whether diagonal \emph{push} constraints make the model more expressive,
and, in particular, whether the stack can be untimed in this case.
(This potentially more general model would still be subsumed by our orbit-finite \trPDA from Sec.~\ref{sec:oftrPDA}.)

\begin{figure}
	
		\begin{tikzpicture}[x=1.5cm]
	
			\draw[->,thick,>=latex] (0,0) -- (5,0) node [below right] {$\scriptstyle \text{time}$};
			\draw[thick] (1,0) -- ++(0, 5pt) node [below = 1ex] {$\push\alpha$};
			\draw[thick] (2,0) -- ++(0, 5pt);
			\draw[thick] (4,0) -- ++(0, 5pt) node [below = 1ex] {$\begin{array}{cc}\text{corresponding} \\ \pop{\alpha \models \psi}\end{array}$};

			\fill[pattern=dots]([xshift=5pt]1,0.27) rectangle node[above] {$k$} ([xshift=-5pt]2,0.37);
			\fill([xshift=5pt]2,0.25) rectangle node[above] {must reset $x$ here} ([xshift=-5pt]4,0.35);
		
		\end{tikzpicture}
		
		

\ignore{
	\\[4ex]
	\begin{subfigure}[b]{\linewidth}
		\centering
		\begin{tikzpicture}[x=1.5cm]
	
			\draw[->,thick,>=latex] (0,0) -- (5,0) node [below right] {$\scriptstyle \text{time}$};
			\draw[thick] (0,0) -- ++(0, 5pt);
			\draw[thick] (1,0) -- ++(0, 5pt) node [below = 1ex] {$\push\alpha$};
			\draw[thick] (4,0) -- ++(0, 5pt) node [below = 1ex] {$\begin{array}{cc}\text{corresponding} \\ \pop{\alpha \models \psi}\end{array}$};

			\fill[pattern=dots]([xshift=5pt]0,0.25) rectangle node[above] {$-k$} ([xshift=-5pt]1,0.35); 
			\fill([xshift=5pt]0,0.85) rectangle node[above] {must reset $x$ here} ([xshift=-5pt]4,.95);  

		\end{tikzpicture}
		\caption{Case $k\leq0$.}
		
	\end{subfigure}
}

	\caption{Reset obligation on $x$ when $z - x \gtrsim k \in \psi$.}
	\label{fig:reset_obligation}
	
\end{figure}
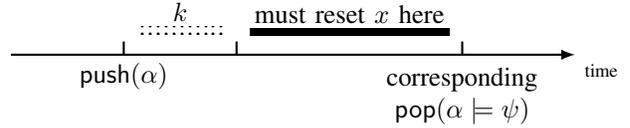

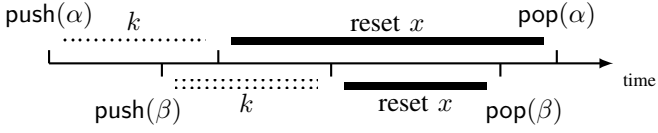
\begin{figure}
	
	\centering	
	\begin{tikzpicture}[x=1.5cm]

		\draw[->,thick,>=latex] (0,0) -- (5,0) node [below right] {$\scriptstyle \text{time}$};
		\draw[thick] (0,0) -- ++(0, 5pt) node [above = 1ex] {$\push\alpha$};
		\draw[thick] (1,0) -- ++(0, -5pt) node [below = 1ex, xshift = -2ex] {$\push\beta$};
		\draw[thick] (1.5,0) -- ++(0, 5pt);
		\draw[thick] (2.5,0) -- ++(0, -5pt);
		\draw[thick] (4,0) -- ++(0, -5pt) node [below = 1ex, xshift = 2ex] {$\pop{\beta}$};
		\draw[thick] (4.5,0) -- ++(0, 5pt) node [above = 1ex] {$\pop{\alpha}$};

		\fill[pattern=dots]([xshift=5pt]0,0.27) rectangle node[above] {$k$} ([xshift=-5pt]1.5,0.37);
		\fill([xshift=5pt]1.5,0.25) rectangle node[above] {reset $x$} ([xshift=-5pt]4.5,0.35); 

		\fill[pattern=dots]([xshift=5pt]1,-.22) rectangle node[below] {$k$} ([xshift=-5pt]2.5,-.33); 
		\fill([xshift=5pt]2.5,-.25) rectangle node[below] {reset $x$} ([xshift=-5pt]4,-.35); 
	
	\end{tikzpicture}
		
	\caption{New reset obligations always subsume current ones.}
	\label{fig:reset_obligation_subsumption}
	
\end{figure}

\ignore{
\begin{figure}
	
		\begin{tikzpicture}[x=1.5cm]
	
			\draw[->,thick,>=latex] (0,0) -- (5,0) node [below right] {$\scriptstyle \text{time}$};
			\draw[thick] (1,0) -- ++(0, 5pt) node [below = 1ex] {$\push\alpha$};
			\draw[thick] (3,0) -- ++(0, 5pt) node [below = 1ex]
				{$\begin{array}{cc}\text{corresponding} \\ \pop{\alpha \models \psi}\end{array}$};			
				
			\fill([xshift=5pt]1,0.25) rectangle node[above] {$x$ reset either here} ([xshift=-5pt]3,0.35); 
			\fill([xshift=5pt]3,0.25) rectangle node[above] {or here} ([xshift=-5pt]5,0.35);  
		
		\end{tikzpicture}
		
		

	\caption{Guess when $x$ is reset.}
	\label{fig:reset_obligation_guess}
	
\end{figure}
}

\ignore{
\begin{remark}
	The construction has exponential complexity in the number of clocks of $\taut$
	due to the guessing of which clocks cannot/must be reset.
	However, the number of clocks of $\uaut$ is \emph{linear} in the size of $\taut$.
	Since regionisation is exponential only in the number of clocks,
	this gives an overall EXPTIME procedure for non-emptiness.
\end{remark}
}


\section{Definable sets and relations}

\label{sec:defin}

In order to go beyond the recognizing power of \dtPDA,
we define automata that use timed registers instead of clocks.
While a clock stores the difference between the current time and the time of its last reset,
a timed register stores an absolute timestamp.
Unlike ordinary clocks, timed registers are suitable for the modeling of non-monotonic time,
and, even in the monotonic setting,
they are more expressive since they can be manipulated with greater freedom than clocks.
While in the semantics of clocks diagonal and non-diagonal constraints are inter-reducible \cite{BerardDiekertGastinPetit:1998:Epsilon},
in the setting of timed registers only diagonal constraints are meaningful.
Consequently, we drop non-diagonal constraints of the form $x \sim k$,
and we redefine the notion of \emph{constraint},
by which we mean a positive boolean combination of formulas 
$x - y \sim k$, where $x, y$ are variables, $k \in \Z$, and $\sim \ \in \{ <, \leq, \geq, > \}$.
We use $=$ and $\neq$ as syntactic sugar.
Constraints are expressively equivalent to the quantifier-free language of the structure $(\Q, \leq, +1)$,
for instance
\[
(x +1 \leq y+1 \ \land \ y \leq x) \ \lor \ \neg (x \leq (y + 1) + 1)
\] 
can be rewritten as a constraint
$x = y \ \lor \ x -y > 2$.
For complexity estimations
we assume that the integer constants are encoded in binary.
%

A constraint $\phi$ over variables $x_1, \ldots, x_n$ defines a subset $\defin{\phi} \subseteq \Q^n$, assuming
an implicit order on variables; $n$ is called the \emph{dimension} of $\phi$, or $\dim{\phi}$. 
In the sequel, we use disjoint unions of sets defined by constraints,
and we call these sets \emph{definable sets}.
Formally, a {definable} set is an indexed set
\begin{align} \label{eq:def set}
X \quad = \quad \{ X_l \}_{l \in L},
\end{align}
where $L$ is a finite index set and for every $l \in L$, the set $X_l = \defin{\phi_l}$ is defined by a constraint.
$(L, \{\phi_l\}_{l \in L})$ is a \emph{constraint representation} of the set~\eqref{eq:def set}.
When convenient we identify $X$ with the disjoint union
\begin{align}  \label{eq:disjoint union}
\biguplus_{l \in L} \ X_l,
\end{align}
and write $\elt{l, v} \in X$ instead of $v \in X_l$.
The automata in this paper will have definable state spaces.
An index  $l \in L$ may be understood as a control location, and a tuple $v \in \Q^n$
may be understood as a valuation of $n$ registers (hence variables may be understood as register names). 
Under this intuition, $\phi_l$ is an invariant that constraints register valuations in a control location $l$.
Similarly, an alphabet letter will contain an element of a finite set $L$,
and a tuple $v \in \Q^n$ conforming to a constraint.

We do not assume that all component sets  $\defin{\phi_l}$ have the same dimension
(in particular, the number of registers may vary from one control location to another).
Observe that when all dimensions are 0, the set~\eqref{eq:disjoint union} is a finite set of the same cardinality as the indexing set $L$. Those sets, as well their elements, we call \emph{timeless}; the elements which are not timeless we call
\emph{\timed}.
When describing concrete definable sets we will omit the formal indexing;
for instance, we will write
\[
\{ l, l', l'' \} \ \uplus \ \Q^2  \qquad \text{ or } \qquad
\{ l, l', l'' \} \ \cup \ \{k \} \times \Q^2
\]
for a set consisting of $\Q^2$ and three other elements.

Along the same lines we define definable (binary) relations. 
For two definable sets
$X = \set{ X_l }_{l \in L}$ and
$Y = \set{ Y_k }_{k \in K}$,
a definable relation $R \subseteq X \times Y$ is an indexed set
\begin{align} \label{eq:fodef rel} 
R \quad = \quad \{ R_{(l, k)} \}_{(l, k) \in L\times K},
\end{align}
where the indexing set is the Cartesian product $L \times K$
and every set $R_{(l, k)}$ is defined by a constraint
satisfying $R_{(l, k)} \subseteq X_l \times Y_k$;
in particular, $\dim R_{(l, k)} = \dim X_l + \dim Y_k$.

Transition relations of automata will be definable relations in the sequel.
The relation $R_{(l, k)}$ is a constraint on
a transition from a control location $l$ to another control location $k$: it prescribes
how a valuation of registers in $l$ \emph{before} the transition may relate to a valuation of registers in $k$
\emph{after} the transition.

Likewise one defines relations of greater arities. Thus a constraint representation of an $n$-ary definable relation
consists of $n$ finite index sets $L_1, \ldots, L_n$, and formulas 
\[
\phi_{(l_1, \ldots, l_n)}, \quad \text{ for } (l_1, \ldots, l_n) \in L_1 \times \ldots \times L_n.
\]
Note that the number of indexes $(l_1,  \ldots, l_n)$ may be exponential in $n$.
When such a relation is input to an algorithm, the presentation is allowed to omit those
formulas which define the empty set, i.e., $\defin{\phi_{(l_1, \ldots, l_n)}} =  \emptyset$.


\vspace{1mm}
\begin{remark}
Constraints are as expressive as first-order logic of $(\Q, \leq, +1)$:
similarly like a constraint, a first-order formula $\phi$ with free variables $x_1, \ldots, x_n$
\emph{defines} the subset $\defin{\phi} \subseteq \Q^n$, and
may be effectively transformed into an equivalent constraint $\psi$, namely one
satisfying $\defin{\phi} = \defin{\psi}$.
%
\end{remark}

\subsection{Orbit-finite sets}

The setting of definable sets is a natural specialisation of the more general setting of 
\emph{sets with timed atoms $(\Q, \leq, +1)$}.
A \emph{time automorphism}, i.e., an automorphism of timed atoms $(\Q, \leq, +1)$,
is a monotonic bijection $\pi : \Q \to \Q$ preserving integer distances, i.e., $\pi(x + k) = \pi(x) + k$ for every $k \in \Z$.
We consider only sets invariant under time automorphism,
which are called \emph{equivariant sets}\footnote{For well-behaved atoms, like equality atoms, 
finitely supported sets can be considered. In case of timed atoms, we restrict ourself to equivariant sets,
i.e., those which are supported by the empty set.}.
In general, equivariant sets are infinite unions of orbits,
where the \emph{orbit} of an element $e$ is
\[
\orbit{e} \ = \  \{ \pi \, e \ : \ \pi \text{ time automorphism} \}.
\]
We restrict our attention to \emph{orbit-finite} sets,
which are those equivariant sets that decompose into a \emph{finite} union of orbits.

A time automorphism $\pi$ acts on any element $e$ by renaming all 
time values $t \in \Q$ appearing in $e$, but leaves
the other structure of $e$ intact. For instance, it distributes on tuples and disjoint unions:
\begin{align*}
	\pi(t_1, \ldots, t_n) \ &= \ (\pi(t_1), \ldots, \pi(t_n)) \\
	\pi(\elt{l, v}) \ & = \ \elt{l, \pi(v)}.
\end{align*}
Thus components $X_l$'s of a definable set are preserved by time automorphisms, and independently partition into orbits.

As an example, the orbit of $(2, 3.3, -1.7)$ is the set
$\set{ (x, y, x') \in \Q^3} {1 < y-x < 2 \ \land \ y - x' = 5}$.
The set $\Q$ has only one orbit
(i.e., $\orbit{x} = \Q$ for every $x \in \Q$),
but the set $\Q^2$ is already orbit-infinite,
its orbits being of the form
\[
\{ (x, y) \ : \ x - y = z\}
\quad \text{ or } \quad
\{ (x, y) \ : \ z < x - y < z+1 \}
\]
for every integer $z \in \Z$. 
Orbits in $\Q^n$ are definable by constraints; we call these constraints \emph{minimal} as
they define the inclusion-minimal nonempty equivariant subsets of $\Q^n$. 
Consequently, every orbit-finite subset of $\Q^n$ is definable. Further, every definable set is equivariant.
On the other hand not every equivariant subset of $\Q^n$ is definable
(e.g., the equivariant set $\set{ (x, y) : x - y \text{ is a prime number}}$),
and not every definable subset is orbit-finite,
due to the orbit-infiniteness of $\Q^n$~\cite{BL12icalp}.

In the sequel whenever we consider an orbit-finite set, we implicitly assume that it is a disjoint union 
of subsets of $\Q^n$, $n \geq 0$, and therefore definable.

Define the \emph{span} of a tuple $v \in \Q^n$ with $n > 0$ as $\max v - \min v$,
the difference between the maximal and the minimal value in $v$,
and for $n = 0$, let the span be $0$ by convention.
\begin{restatable}{lemma}{BoundedSpanLemma} \label{lem:bounded-of}
	An equivariant subset $X \subseteq \Q^n$ is orbit-finite if, and only if, it has uniformly bounded span,
	i.e., it admits a common bound on the spans of all its elements.
\end{restatable}
For an orbit $O \subseteq X$,
we will make the notational difference between \emph{the orbit of} $e$, when $O = \orbit{e}$ with $e \in X$,
and \emph{an orbit in} 
$X$, when $O \in \orbit X := \setof {\orbit e} {e \in X}$.

 
\subsection{Normal form}

We prove that every definable set can be transformed into a convenient normal form,
which is like the classical partitioning into regions but without restricting to non-negative rationals.

We say that a tuple $v \in \Q^d$ admits a \emph{gap} $g \in \Q$, $g > 0$, if the set of rationals  
appearing in $v$ can be split into two non-empty sets $L, R \subseteq \Q$ such that 
\[
\max(L) + g = \min(R).
\]
Let a \emph{$g$-extension} of $v$ be any tuple in $\Q^d$ obtained from $v$ by
adding a positive rational $h$  to all elements of $R$ appearing in $v$,
regardless of the choice of sets $L$ and $R$.
(Subtracting $h$ from all elements of $L$ would be equivalent for our purposes 
as the sets we consider are closed under translations.)

Note that 
if $v$ admits an integer gap $k$ then all other tuples in $\orbit{v}$ also do.
If this is the case for some $k \in \Z$,
let the $k$-extension of an orbit $O \subseteq \Q^d$ be the closure of $O$ under $k$-extensions, i.e., the smallest set
containing $O$ that contains all $k$-extensions of all its elements. 
We will build on the specific weakness of constraints: a fixed constraint can not distinguish an orbit $O$ from its $k$-extension 
when $k$ is sufficiently large.

An extension (i.e., a $k$-extension for some integer $k$) of an orbit $O$ is a definable set.
Indeed, a defining constraint is obtained from the minimal constraint $\phi$ defining $O$ by 
syntactically replacing certain equalities $=$ with inequalities $\leq$ (call this constraint extension of $\phi$).
For instance, consider 
\[
\phi(x, y, z, w) \ \equiv \  0 < y -x < 1 \ \land \ z-y = 7 \ \land \  w -z = 7;
\]
its $7$-extension is $0 < y -x < 1 \ \land \ z-y \leq 7 \ \land \  w -z \leq 7$.

\begin{restatable}[Normal Form Lemma]{lemma}{NormalFormLemma}   \label{lem:nf}
Every definable set $X$ decomposes into a finite union of orbits $O \subseteq X$
and of extensions of orbits $O \subseteq X$. A decomposition can be effectively computed in \exptime.
\end{restatable}
\noindent
For orbit-finite $X$, the lemma yields an effective enumeration of orbits in $X$, since extensions of orbits,
being orbit-infinite, do not appear
in the decomposition of $X$:
\begin{corollary} \label{cor:nf}
A decomposition of an orbit-finite definable set $X$ into orbits is computable in \exptime.
\end{corollary}

\begin{example}
	Consider the following set 
	$X = \setof{ (x, y, z) \in \Q^3}{0 < y-x < 1 \land z-y > 3}$.
	One possible decomposition of
	$X$ consists of orbits in $X$ that do not admit a gap larger than $4$, and of $4$-extensions of all orbits that
	admit a gap $4$, for instance the $4$-extension of the orbit
	$\{ (x, y, z) \in \Q^3 \ : \ 0 < y-x < 1 \ \land \  z-y = 4  \}$.
\end{example}

\noindent
Thanks to the Normal Form Lemma, we define the \emph{normal form} of constraints, i.e., disjunction of 
minimal constraints and extensions thereof. 
In the sequel we assume whenever convenient that the constraint representations of definable sets 
are already in normal form.
The exponential blowup introduced by this transformation will combine well with the polynomial complexity w.r.t.~the
normal form representations, thus yielding the exponential time overall complexity. 

A relevant property of normal form sets is that they admit easy computation of projections:
\begin{lemma}[Projection Lemma] \label{lem:proj}
	Given a definable set $X \subseteq \Q^d$ in normal form, its projection 
	onto a subset of coordinates $\{1 \ldots d\}$, in normal form, is computable in polynomial time.
\end{lemma}
\noindent
Indeed, projection distributes over disjunction, and projection of
a minimal constraint, or of extension thereof, is computed essentially by elimination of variables.


\section{\Tr pushdown automata}  \label{sec:trPDA}


\ignore{
	As a preparation consider the classical definition of NFA consisting of a finite input alphabet $A$, finite set of states $Q$, initial and 
	final states $I, F \subseteq Q$, and a transition relation $\delta \subseteq Q \times A \times Q$.
	Let's reinterpret (and generalize) this definition by dropping the requirement that the alphabet and 
	states are finite sets;
	instead, we require $Q$, $A$, $I$ and $F$ to be orbit-finite definable sets, and $\delta$ to be 
	a definable relation\footnote{Assuming the transition relation to be orbit-finite, i.e., 
	of uniformly bounded span, would be too restrictive, as it would impose a bound on
	change of time values stored in the state, in one transition.}
	Call this generalized model \emph{\tr NFA}.

	The semantics of a time-register NFA is defined as usual, including runs, accepting runs, and the language
	$L \subseteq A^*$ recognized by an automaton. 
	We often write $\trans{q}{a}{q'}$ instead of $(q, a, q') \in \delta$.
	As usual, we will allow for epsilon-transitions in NFA, i.e., transitions (labeled by a special symbol $\eps \in A$) 
	which do not read an input letter.

	How do \tr NFA relate to classical timed automata~\cite{AD94}?
	A timed automaton stores in a clock the difference between the current time value 
	and some past one (the time when the clock was last reset),
	and compares the difference to an integer constant;
	a \tr automaton stores just a time value in a register, but likewise compares a difference to a constant.
	If one assumes timed automata to have uninitialised clocks (see~\cite{BL12icalp} for details),
	timed automata are easily shown to correspond to a subclass of \tr NFA, where: 
	the input alphabet is 1-dimensional; only monotonic input words are considered,
	i.e., input time values appear in non-decreasing order;
	all states have the same dimension; and registers may be updated only with the current input time value.
}


We define a new model of timed PDA by reinterpreting the standard presentation of PDA in the setting of 
definable sets. 
Our approach generalized the approach of~\cite{BL12icalp} where NFA were considered.
Classical PDA can be defined in a number of equivalent ways. 
In the setting of this paper, the choice of definition will be critical for tractability.
In the most general variant, a PDA $\aaut$ consists of a finite input alphabet $A$,
a finite set of states $Q$, initial and final states $I, F \subseteq Q$, 
a finite stack alphabet $S$, and a finite set of transition rules 
\begin{align*}
	\rho \ \subseteq \ (Q \times S^*) \times A_\varepsilon \times (Q \times S^*),
\end{align*}
where $A_\varepsilon = A \cup \set{\varepsilon}$.
The semantics of a PDA is defined as usual.
A transition rule $(q, v, a, q', v') \in \rho$ describes a transition which reads input $a$,
changes state from $q$ to $q'$, pops a sequence of symbols $v$ from the stack and replaces it by $v'$.
Formally, the transitions of a PDA are between configurations $c, c' \in Q \times S^*$, and 
$(q, v, a, q', v') \in \rho$ induces a transition $\trans{c}{a}{c'}$ if 
$c = (q, v w)$ and $c' = (q', v' w)$ for some $w \in S^*$.
Similarly, one defines unlabeled transitions $\trans{c}{}{c'}$,  
the reachability relation $\transtrans{c}{}{c'}$,
runs, accepting runs (runs starting in a state from $I$ with empty stack,
and ending in a state from $F$ with arbitrary stack), 
and the language $L(\aaut)\subseteq A^*$ accepted by a PDA $\aaut$.

We reinterpret the definition of PDA by dropping the finiteness of the components.
Instead, we require $Q, A, S, I$ and $F$ to be orbit-finite (and, thus, definable), and the relation~$\rho$ to be definable.
The \emph{dimension} of a PDA is the maximal dimension of its states $Q$.
These orbit-finiteness requirements are necessary to obtain a model with decidable emptiness,
since it has been shown in \cite{BL12icalp} that having orbit-infinite states leads to undecidability already in NFA.
Since $Q$ is orbit-finite, by Lemma~\ref{lem:bounded-of} there exists a uniform bound 
on the span of every vector in $Q$.

Note that $\rho$, being definable, is necessarily a subset of 
\begin{align} \label{eq:rules m n}
	\rho \ \subseteq \ (Q \times S^{\leq n}) \times A_\varepsilon \times (Q \times S^{\leq m}),
\end{align}
for some $n, m \in\nat$, where $S^{\leq n} = S_0 \cup S \cup S^2 \cup \ldots \cup S^n$,
where $S_0 = \{\ew\}$.
The generalized model we call \emph{\tr PDA} (\trPDA).
Most importantly, the semantics of \trPDA is defined \emph{exactly as} the semantics of classical PDA.
We assume acceptance by final state.
This is expressively equivalent to acceptance by empty stack, or by final state and empty stack.

By the \emph{size} of a \trPDA we mean the size of its constraint representation, i.e., the
sum of sizes of all defining constraints, where we assume that integer constants are encoded in binary.

As already in the case of NFA, also for PDA imposing an orbit-finiteness restriction on $\rho$ would be too restrictive,
in the sense that the model would recognize a strictly smaller class of timed languages than with unrestricted $\rho$.
Example~\ref{ex:pda1} illustrates this, and shows the interaction between 
timed symbols in the stack, state, and input.
\begin{example}
	\label{ex:pda1}
	Consider the input alphabet $A = \defin{\phi}$, where
	$\phi(x, y)  \ \equiv \ x < y < x+4$,
	and the language $L$ of even-length monotonic palindromes over $A$,
	i.e., $L = \{(u_1, v_1) \ldots (u_{2n}, v_{2n}) \in A^* \ | \ 
	u_1 \leq \ldots \leq u_n \textrm{ and } (u_i, v_i) = (u_j, v_j) \textrm{ for every } 1 \leq i \leq 2n \textrm{ and } j = 2n + 1 - i\}$.
	A \trPDA recognizing this language has state space of dimension 1 (i.e., 1 register)
	$Q = \{ i \} \uplus \{ 1 \} \times \Q \uplus \{2, f\}$,
	with $i$ and $f$ the initial and final states, respectively.
	The stack alphabet $S = A \uplus \{\bot\}$ extends the input alphabet by the symbol $\bot$.
	There are three groups of $\varepsilon$-transition rules, namely
	\[
	(i, \varepsilon, \varepsilon, (1, t), \bot), \qquad ((1, t), \varepsilon, \varepsilon, 2, \varepsilon), \qquad (2, \bot, \varepsilon, f, \varepsilon),
	\]
	for any $t \in \Q$,
	used to initiate the first half, to change to the second half, and to finalize the second half of 
	a computation of an automaton. 
	In state $(1, t)$ the automaton pushes an input letter $(u, v)$ to the stack,
	while checking for monotonicity $t \leq u$, as described by the transition rules, for $t \leq u$,
	\[
	((1, t), \varepsilon, (u, v), (1, u), (u, v)) \in Q \times S^0 \times A \times Q \times S. 
	\]
	Finally, in state $2$ the automaton pops a symbol $(u, v)$ from the stack, while checking for equality with
	the input letter, as described by the transition rules:
	\[
	(2, (u, v), (u, v), 2, \varepsilon) \in  Q \times S \times A \times Q \times S^0. 
	\]
	Observe that we can not require the set $\rho$ of transition rules to be orbit finite. Indeed, this  would impose a
	bound on the span of tuples in $\rho$, in particular on the difference $u - t$ in the push transition rules, and therefore
	also on the differences $u_{i+1} - u_i$ between consecutive input letters.
\end{example}
%
%


The \emph{non-emptiness problem} asks
whether the language recognised by a given \trPDA is non-empty.
We observe that the problem is undecidable for general \trPDA.
\begin{restatable}{theorem}{ThmUndecid} \label{thm:undecid}
	Non-emptiness of \trPDA is undecidable.
\end{restatable}
The undecidability of the general model motivates us to consider several restrictions of \trPDA
for which we can show decidability of the non-emptiness problem.
We consider \tr context-free grammars in Sec.~\ref{sec:CFG},
orbit-finite \trPDA in Sec.~\ref{sec:oftrPDA},
and \trPDA with timeless stack in Sec.~\ref{sec:timelesstrPDA}.


\subsection{\Tr context-free grammars}

\label{sec:CFG}


Context-free grammars are PDA with one state
where each transition pops exactly one symbol off the stack.
A \emph{\tr context-free grammar} (\trCFG) $\g$ consists of the following items:
an orbit-finite set $S$ of symbols,
a starting symbol $I \in S$ which is initially pushed on the stack,
an orbit-finite input alphabet $A$,
and a definable set of productions
\[
\rho \ \subseteq \ S \times A_\varepsilon \times S^*.
\]
Acceptance is by empty stack, i.e., when all symbols are popped off the stack.
We call languages recognized by \trCFG \emph{\tr context-free languages}.

\begin{example}
	\label{ex:timed_palindromes}
	Let $A = \Q$, and consider the language $L$ of timed palindromes of even length,
	i.e., $L = \setof{x_1 \cdots x_{2n} \in \Q^*} {\forall (1 \leq i \leq n) \cdot x_i = x_{2n - i + 1}}$.
	This language can be recognized by a \trCFG with symbols $S = \set 1 \uplus \set 2 \times \Q$
	and productions of the form
	$\rho = \setof{(1, x, 1 \cdot (2, y)), (1, x, (2, y)), ((2,x), y, \varepsilon)} {x, y \in \Q \cdot x = y}$.
	We will see later that this language cannot be accepted by \trPDA with timeless stack.
\end{example}

Define the \emph{untiming} of a word $a_1 \ldots a_n \in A^*$ over an orbit-finite alphabet $A$ 
as its projection to orbits
$\orbit{a_1} \ldots \orbit{a_n} \in \Sigma^*$,
where $\Sigma = \text{orbits}(A)$.
Untiming naturally extends to languages.
In the lemma below we show that the untiming of a language of \trCFG is context-free.
This contrasts with languages of \trPDA; cf. Example~\ref{ex:pda2}.
Therefore, \trCFG are weaker than general \trPDA.

\begin{restatable}{lemma}{LemtrCFL}
	\label{lem:trCFL}
	The untiming of a \tr context-free language is effectively context-free.
\end{restatable}
%
%

\begin{restatable}{theorem}{ThmCFG}
	\label{thm:trCFG}
	Non-emptiness problem of \trCFG is \exptime-complete.
\end{restatable}
%


\subsection{Orbit-finite \tr PDA}

\label{sec:oftrPDA}

We have seen that restricting \trPDA to grammars yields a decidable model.
In this section, we investigate another natural restriction of \trPDA with decidable non-emptiness.
A transition rule $(q, v, a, q', v') \in \rho$ splits naturally into its left-hand side (lhs) $(q, v) \in Q\times S^*$ 
and its right-hand side (rhs) $(q', v') \in Q\times S^*$.
Let \emph{orbit-finite \trPDA} be the subclass of \trPDA where the projections 
of $\rho$ to both lhs's and rhs's, i.e., the following two sets
\begin{align*}
&\setof{ (q, v)}{\exists a, q', v' . \ (q, v, a, q', v') \in \rho } \\ 
&\setof{ (q', v')}{\exists q, v, a. \ (q, v, a, q', v') \in \rho },
\end{align*}
are orbit-finite.
By Lemma~\ref{lem:bounded-of} this means that both lhs's and rhs's have uniformly bounded span.
We still  \emph{do not} require the whole relation $\rho$ to be orbit-finite.


As long as the recognized language is considered,
orbit-finite \trPDA may be transformed into a convenient short form, with
the transition rules split into
\begin{align} \label{eq:short}
\begin{aligned}
\rho & \ = \ \rhopush \ \cup \ \rhopop, \\ 
\rhopush & \ \subseteq \ Q \times A_\varepsilon \times Q \times S, \\ 
\rhopop & \ \subseteq \ Q \times S \times A_\varepsilon \times Q
\end{aligned}
\end{align}
(thus one of lhs, rhs is a single state from $Q$, and the other is a pair from $Q \times S$)
where the two sets
\begin{align*}
& \setof{ (q', s')}{ \exists q, a . \ \rhopush(q, a, q', s') } \\
& \setof{ (q, s)}{ \exists q, a . \ \rhopop(q, s, a, q') } 
\end{align*}
are orbit-finite.
This short form easily enables the simulation of transition rules of the form 
$\rhonop(q, a, q') \in Q\times A_\varepsilon \times Q$
that do not operate on stack, by a push followed by a pop.
The \trPDA in Example~\ref{ex:pda1} is in short form.

\begin{restatable}{lemma}{LemoftrPDA}  \label{lem:oftrPDA}
An orbit-finite \trPDA can be transformed into a language-equivalent \trPDA in short form~\eqref{eq:short}
of polynomially larger size.
\end{restatable}
\noindent
Thus, from now on we always conveniently assume that an orbit-finite \trPDA is given in short form.
According to the following example, untiming of the language of an orbit-finite \trPDA needs not be context-free.
\begin{example}
	\label{ex:pda2}
	Consider the language $L$ of palindromes over the timeless alphabet $A = \set{a, b}$ 
	containing the same number of $a$'s and $b$'s.
	$L$ can be recognized by a \trPDA of dimension 1 
	with state space
	$Q = \{ i \} \ \uplus \ \{ 1,2  \} \times \Q \ \uplus \ \{f\}$
	and stack alphabet $S = \{a, b\} \ \uplus \ \{\bot\} \times \Q$.
	as follows.
	At the beginning, a rational $t \in \Q$ is guessed and $(\bot, t)$ is immediately pushed 
	to the stack according to the transition rules:
	\[
	(i, \varepsilon, (1, t), (\bot, t)) \in Q \times \{\varepsilon\} \times Q \times S, \qquad \text{ for } t \in \Q.
	\]
	Palindromicity of $L$ is checked by pushing timeless symbols $a, b$ on the stack in the first half of the computation,
	and by popping and matching them during the second half.
	Additionally, the value stored in the state is increased at each occurrence of $a$, and decreased at each occurrence of $b$,
	according to the transition rules:
	\begin{align*}
	\setof{\begin{array}{l}
	 ((1, t), a, (1, t+1), a), \\
	 ((1, t), b, (1, t-1), b)
	 \end{array}}{\ t \in \Q\ }
	  \ \subseteq \ Q \times A \times Q \times S \\
	\setof{\begin{array}{l}
	((2, t), a, a, (2, t+1)), \\
	 ((2, t), b, b, (2, t-1)) 
	 \end{array}}{\ t \in \Q\ }
	\  \subseteq \ Q \times S \times A \times Q 
	\end{align*}
	At the end of the computation, it remains to check that the number of $a$'s equals the number of $b$'s.
	After the last timeless symbol is popped off the stack,
	on the  bottom thereof we have $(\bot, t)$ where $t$ is the original value stored there at the beginning of the computation.
	It suffices to pop this timed symbol with a transition rule:
	\[
	((2, t), (\bot, t), \varepsilon, f) \ \in \ Q\times S \times \{\varepsilon\} \times Q, \qquad \text{ for } t\in \Q,
	\]
	which checks equality with the value stored in the state.
\end{example}

As our second main result, we prove decidability of non-emptiness for orbit-finite \trPDA:
\begin{theorem} \label{thm:nexptime}
	Non-emptiness of orbit-finite \trPDA is in \nexptime.
\end{theorem}
\noindent
Recall that we assume that integer constants appearing in constraint representation of a \trPDA are encoded in binary.
We prove the theorem in Sec.~\ref{sec:upper-bounds} by reducing non-emptiness of \trPDA
to non-emptiness of systems of equations over set of integers.


\subsection{\trPDA with timeless stack}

\label{sec:timelesstrPDA}

To obtain a better complexity upper-bound, and for comparison with previous work,
we identify the subclass of \trPDA where the stack alphabet is timeless 
(i.e., finite). We call this subclass \emph{trPDA with timeless stack},
which corresponds precisely to timed-register automata \cite{BL12icalp}
augmented with a timeless stack (in the spirit of \cite{bouajjani:timed:PDA:94}).
Observe that this is a subclass
of orbit-finite \trPDA, by the following observation:
\begin{proposition} \label{prop:product orbit finite}
	Cartesian product of an orbit-finite set and a timeless one is orbit-finite.
\end{proposition}
\noindent
Thus, lhs and rhs are orbit-finite if $Q$ is orbit-finite and $S$ is timeless.
This class is weaker than orbit-finite \trPDA.
Indeed, the automaton recognizing language $L$ described in Example~\ref{ex:pda2}  is orbit-finite.
On the other hand $L$ is not recognized by a \trPDA with timeless stack, due to the following:
%
%
\begin{lemma} \label{lem:untiming}
Untiming of the language of \trPDA with timeless stack is effectively context-free.	
\end{lemma}
\begin{proof}[Proof (sketch).]
	Replace the state space $Q$ by the set of orbits of $Q$ (similarly to the region construction), and consider
	transitions between orbits, labelled with orbits of the input alphabet $A$, defined existentially.
	This operation does not preserve the timed language $L$ recognized by the automaton in general,
	but it does preserve reachability properties,	
	and in particular the untiming projection of $L$.
	Since the stack is timeless, no special care is needed to handle it.
\end{proof}
Languages of \trPDA with timeless stack are thus a strict subclass of those of orbit-finite \trPDA,
even over finite alphabets.
Moreover, languages of \trCFG are incomparable with languages of \trPDA with timeless stack.
An example of \trCFG language which is not recognized by \trPDA with timeless stack
is the language of timed palindromes from Example~\ref{ex:timed_palindromes}.
This language clearly cannot be recognized with a timeless stack
since it requires to remember unboundedly many possibly different timestamps.
For the other inclusion,
the example below shows a language recognized by a \trPDA with timeless stack
but not recognized by a \trCFG.
\begin{example}
	Take $A = \set c \times \Q \uplus \set {a, b}$,
	and consider the language
	$$L = \setof {(c, x) \, w \, (c, y)} {w \textrm{ palindrome over } \set{a, b}, y - x = \len w}.$$
	$L$ can be recognized by a \trPDA with timeless stack which stores $x$ in a register,
	and then uses the untimed stack to check that $w$ is a palindrome and incrementing the register at every letter.
	Finally, it checks that $y$ equals the value of the register.
	It can be shown that $L$ cannot be recognized by a \trCFG by a standard pumping argument. 
	Intuitively, a sufficiently long word $s \in L$
	can be split into $s = uvwxy$ s.t. at least one of $v$ and $x$ is non-empty,
	and, for every $i \geq 0$, $s_i := uv^iwx^iy \in L$.
	Since $s$ has only two timestamps (at the beginning and at the end),
	pumping cannot involve them.
	Thus, $v$ and $x$ are substrings of the palindrome $w$,
	and pumping necessarily changes its length, which contradicts $s_i \in L$.
\end{example}

As our last main result, we derive a tight upper complexity bound for \trPDA with timeless stack.
\begin{theorem} \label{thm:exptime}
	Non-emptiness for \trPDA with timeless stack is \exptime-complete.
\end{theorem}

\begin{remark}
	It follows from the proof that non-emptiness of automata in normal form
	is decidable in time polynomial in its size	and exponential in its dimension.
\end{remark}



\subsection{\dtPDA as \trPDA with timeless stack}

Our definition of \trPDA differs from \dtPDA~\cite{AbdullaAtigStenman:DensePDA:12}
in the same way as timed register automata of~\cite{BL12icalp} differ from classical timed automata~\cite{AD94}.
The first difference is semantic: \dtPDA (like timed automata) recognize timed languages where each input symbol carries only a single time-stamp. In this sense, they correspond to \trPDA with a 1-dimensional input alphabet.

Moreover, languages of \trPDA are closed under translations $x \mapsto x+t$, for $t \in \Q$, 
while languages of \dtPDA are not.
In order to fairly compare the two models, we assume (along the lines of~\cite{BL12icalp}) 
that a \dtPDA starts its computation with \emph{uninitialized clocks}, instead of all clocks initialized with $0$. 
This is not a restriction since a \dtPDA $\mathcal T$ can be faithfully simulated by a \dtPDA with uninitialized clocks 
$\mathcal T'$.
For instance, as the first step, $\mathcal T'$ may initialize all its clocks with the timestamp 
of the first input letter $(a, t)$ and then proceed as $\mathcal T$, and thus
$L(\mathcal T') \ = \ \bigcup_{t\in \Q,\, a \in \Sigma} (a, t)\, (L(\mathcal T) + t)$.
This transformation clearly preserves non-emptiness.

\begin{lemma}
	\label{lem:from:DT-PDA:timeless-stack:to:trPDA}
	A \dtPDA with uninitialized clocks and timeless stack $\mathcal T$
	can be effectively transformed into a language-equivalent
	normal form \trPDA $\aaut$ with timeless stack.
	If $\mathcal T$ has $n$ clocks then
	the dimension of $\aaut$ is $n+1$ and its size is exponential in $n$. 
\end{lemma}

We sketch the construction.
By definition, \dtPDA accept monotonic words, while 
languages recognized by 1-dimensional \trPDA are non-monotonic in general.
Notice that monotonicity of input can be enforced by a \trPDA
by adding an additional special register $x_0$ in every control state, to store the timestamp of the last input,
and by intersecting the transition rules with the additional constraint
$x_0 \leq x'_0$
relating the values of the special register before and after a transition.
%

The most substantial difference is that \dtPDA use \emph{clocks}, while \trPDA use \emph{registers}.
A \dtPDA has clocks which can be reset and can be compared to an integer constant $x \sim k$,
or, in the case of diagonal constraints, a difference of clocks is compared to an integer constant $x - y \sim k$.
%
A \trPDA can simulate a \dtPDA by having one register $\hat x$ for each clock $x$.
A reset of $x$ is simulated by assigning the current input timestamp $t$ to $\hat x$;
a constraint $x \sim k$ is simulated by $x_0 - \hat x \sim k$
(where $x_0$ is the special register discussed above),
and a diagonal constraint $x - y \sim k$ is simulated by $\hat y - \hat x \sim k$.
(The ages for timed stack symbols could be treated similarly. 
This step is unnecessary for \dtPDA with timeless stack.)

To obtain a \trPDA we need to ensure that the set of states is orbit-finite.
This is done as follows. Let $m$ be the maximal absolute value of a constant in any constraint of a \dtPDA.
We perform the classical region construction of the \dtPDA, and take regions as control locations of the \trPDA.
In every control location, the defining constraint is the intersection of the region with the constraint
$
\bigwedge_{x \in X} 0 \leq x_0 - x \leq m
$,
which makes the set of states orbit-finite.
Additionally in every region, those registers that correspond to unbounded clocks are projected away. 
This is correct as the truth value of transitions constraints involving unbounded clocks does not depend on 
further elapse of time.
This completes the sketch of the construction claimed in Lemma~\ref{lem:from:DT-PDA:timeless-stack:to:trPDA}.

\ignore{
\begin{example}
	\label{ex:pda}
	Let's extend an automaton from Example~\ref{ex:nfa} with ...

	an orbit-finite \trPDA in short form ... (todo)
\end{example}
}

%

By Theorem~\ref{thm:Abdulla}, we can remove time from the stack of a \dtPDA
with a single exponential blowup in the number of control locations (w.r.t. the size of pop constraints),
and a linear increase in the number of clocks.
By Lemma~\ref{lem:from:DT-PDA:timeless-stack:to:trPDA}, we obtain a \trPDA with a further exponential blowup in the number of control locations (w.r.t. number of clocks).
Notice that the two blowups compose to a single exponential blowup, as summed up in the following corollary:

\begin{corollary} \label{cor:Abdulla}
	A \dtPDA with uninitialized clocks can be effectively transformed into a language-equivalent
	normal form \trPDA with timeless stack of exponential size (w.r.t. pop constraints and clocks) and 
	linear dimension.
\end{corollary}

%

%
In turn, the blowups in the last corollary and in Theorem~\ref{thm:exptime} compose again  
to a single exponential blowup. Therefore Theorem~\ref{thm:exptime} yields 
the \exptime upper-bound for \dtPDA and thus
strengthens the \exptime upper bound of~\cite{AbdullaAtigStenman:DensePDA:12}.
%
%
%
%

 
\section{Upper bounds}  \label{sec:proof}

\label{sec:upper-bounds}

\noindent
We prove the upper bounds of Theorems~\ref{thm:nexptime} and~\ref{thm:exptime}.


\subsection{Equations over sets of integers}

We consider systems of equations, interpreted over sets of integers, of the following form
\begin{align*}
X_1 \ & = \ t_1 \\
& \ldots \\
X_n \ & = \ t_n,
\end{align*}
one for each variable $X_i$,
where right-hand side expressions $t_1, \ldots, t_n$ use variables $X_1 \ldots X_n$ appearing in
left hand sides, constants
$\{-1\}, \{0\},  \{1\}$,
union $\cup$, intersection $t \cap \{0\}$ with the constant $\{0\}$, 
and element-wise addition of sets of integers,
$X + Y \ = \ \{ x+y \ : \ x \in X \text{ and } y \in Y \}$.
Note that the use of intersection is assumed to be very limited; for systems of equations with
unrestricted intersection (e.g., $X \cap Y$), the non-emptiness problems is undecidable~\cite{JezO10}.

A solution $\nu$ of a system of equations assigns to every variable $X$ a set $\nu(X) \subseteq \Z$
of integers.
We are only interested in the least solution.
Note that intersection and addition distribute over union,
in the sense that
$(t_0 \cup t_1) \cap t_2 = (t_0 \cap t_2) \cup (t_1 \cap t_2)$, and
$(t_0 \cup t_1) + t_2 = (t_0 + t_2) \cup (t_1 + t_2)$. Thus, as long as the least
solution is considered, a system of equations may be equivalently presented by a set of inclusions
$X \supseteq t$, where $t$ does not use union, with the proviso that many inclusions may
apply to the same left-hand side variable $X$.

\begin{example}
	\label{ex:equations}
	For instance, the set of all integers is the least solution for $Z$ below;
	we can also succinctly represent large constants $m \in \Z$ as the least solution $\{m\}$ for $Z_{=m}$:
	\begin{align*}
		\begin{aligned}
	 	Z\ &\supseteq \ \{0\} \\
	 	Z\ &\supseteq \ \{1, -1\} + Z
		\end{aligned} 
		\qquad \qquad
	%
		\begin{aligned}
	 	Z_{=0}		\ &\supseteq \ \{0\} \\
		Z_{=2m}		\ &\supseteq \ Z_m + Z_m \\
		Z_{=2m+1}	\ &\supseteq \ Z_m + Z_m + \{1\}.
		\end{aligned}
	\end{align*}
	Infinite intervals of the form $\Z_{<m} = (-\infty, m)$ and $\Z_{>m} = (m, \infty)$
	are easily expressible as the least solutions of $Z_{<m}$ and $Z_{>m}$ in
	\begin{align*}
		\begin{aligned}
			Z_{<m} \ &\supseteq\ Z_{= (m-1)} \\
			Z_{<m} \ &\supseteq\ Z_{<m} + \{-1\}
		\end{aligned}
		\quad \textrm{ and } \quad
		\begin{aligned}
			Z_{>m} \ &\supseteq\ Z_{= (m+1)} \\
			Z_{>m} \ &\supseteq\ Z_{>m} + \{1\}
		\end{aligned}		
	\end{align*}
	We will use these definitions later in this section.
\end{example}

By introducing additional auxiliary variables, one
may easily transform the inclusions into the following binary form:
\begin{align}  \label{eq:int constants}
X \  \supseteq \ \{k\} \qquad
X \  \supseteq \ Y \cap \{0\} \qquad
X \  \supseteq \ Y + Z,
\end{align}
where $k$ is $-1, 0$ or $1$.
For future reference we distinguish a subclass of \emph{intersection-free} systems of equations which
use no intersection. All equations in the previous Example~\ref{ex:equations} are of this form.

The \emph{non-emptiness problem} asks, for a given system $\Delta$ of equations and a variable $X$ therein,
whether the least solution $\nu$ of $\Delta$ assigns to $X$ a non-empty set of integers.
The \emph{membership problem} asks, given an additional integer $k \in \Z$ (coded in binary),
whether $k \in \nu(X)$.

\begin{lemma} \label{lem:setsofeq:PTIME}
	The non-emptiness problem for intersection-free systems of equations is in P.
\end{lemma}
\begin{proof}
	If $\Delta$ is intersection-free, its non-emptiness reduces to non-emptiness
	of a context-free grammar over three letters $\set{-1, 0, 1}$.
	Variables of $\Delta$ are non-terminal symbols, and every inclusion gives raise to
	one production. Addition is replaced by concatenation.
\end{proof}

\begin{restatable}{lemma}{LemsetsofeqNP} \label{lem:setsofeq:NP}
	The non-emptiness and membership problems of systems of equations are both NP-complete.
	The membership problem is NP-hard already for intersection-free systems.
\end{restatable}

\subsection{From \trPDA to systems of equations}

\label{sec:fromtrPDA:to:equations}

%

\noindent
We show an \exptime reduction of non-emptiness of orbit-finite \trPDA
to non-emptiness of systems of equations.
Additionally, if the stack is timeless, then the system of equations is intersection-free.
Fix an orbit-finite \trPDA $\aaut$, with states $Q$, 
stack alphabet $S$, and transition rules $\rhopush$ and $\rhopop$.

As a preprocessing we apply few simplifying transformations.
First, we rebuild $\aaut$ so that it has exactly one (therefore timeless) initial state, and exactly one final state.
Therefore there are unique initial and final control locations, corresponding to the unique timeless initial and final state.
Moreover, in the final state we let $\aaut$ unconditionally pop all symbols from the stack,
and assume w.l.o.g.~that $\aaut$ accepts when not only it is in the final state, but additionally the stack is empty.
%
%
As the next step of preprocessing, we make all states of $\aaut$ timed, by adding to every timeless state
(including the initial and final one) one dummy timed register.
In order to assure orbit-finiteness of $\aaut$, appropriate additional constraints on
the dummy registers are added to $\rhopush$ and $\rhopop$.
Thus the transformations described by now preserve orbit-finiteness of $\aaut$
can be done using its constraint representation.
As the last step of preprocessing, we transform $\aaut$ into normal form.
According to Lemma~\ref{lem:nf} this is doable in \exptime.

\mysubsubsection{Reachability relation}
As we focus on reachability, we ignore the input alphabet and assume the transition rules of $\aaut$ 
to be unlabeled, i.e., of the form
\[
\rhopush(q, q', s') \quad \textrm{ and } \quad
\rhopop(q, s, q'),
\]
where $q, q' \in Q$ and $s' \in S$.
Consequently, we assume also unlabeled transitions $\trans{c}{}{c'}$ between configurations.
Using the Projection Lemma~\ref{lem:proj},
the unlabeled transition rules are easily computed by projecting away the input alphabet.

We define the following binary reachability relation between states of $\aaut$.
Two states are related, written $\horiztrans{q}{}{q'}$,
if there is a computation of $\aaut$ from state $q$ to $q'$ which starts and ends with empty stack.
Formally, $\horiztrans{q}{}{q'}$ if for some configurations $(q_1, v_1), \ldots, (q_n, v_n)$, $\aaut$
admits the transitions: 
%
\begin{align} \label{eq:trans above}
\trans{(q, \ew)}{}{\trans{(q_1, v_1)}{}{\trans{\ldots}{}{\trans{(q_n, v_n)}{}{(q', \ew)}}}}.
\end{align}
It might be the case that $v_i = \varepsilon$ for some $1 \leq i \leq n$.
\begin{proposition} \label{claim:p1}
$L(\aaut)$ is non-empty iff $\horiztrans{}{}{}$ relates an initial state with a final one.
\end{proposition}
%
\begin{lemma} \label{lem:horiz least}
	The relation $\horiztrans{}{}{}$ is the least relation satisfying the following rules:
	\begin{align*}
		&\textrm{\em (base)} &&
			\frac{}{\horiztrans{q}{}{q}}
				&& \forall\ (q, q) \in Q^2 \\[2ex]
		&\textrm{\em (transitivity)} &&
			\frac{\horiztrans{q}{}{q'} \quad \horiztrans{q'}{}{q''}}{\horiztrans{q}{}{q''}}
				&& \forall\ (q, q', q'') \in Q^3 \\[2ex]
		&\textrm{\em (push-pop)} &&
			\frac{\horiztrans{\bar q}{}{\bar q'}}{\horiztrans{q}{}{q'}}
				&& \forall\ (q, \bar q, \bar q', q') \in \pushpop
	\end{align*}
	where $\pushpop$ is the subset of $Q^4$ defined as:
	\begin{align} \label{eq:rule}
		\pushpop = \setof
			{(q, \bar q, \bar q', q')}
			{\exists \bar s \in S. \begin{array}{c} \rhopush(q, \bar q, \bar s), \\  
			\rhopop(\bar q', \bar s, q')\end{array}}
	\end{align}
\end{lemma}

\mysubsubsection{Orbitization}
Recall that the transition rules $\rhopush$ and $\rhopop$ are equivariant, i.e., are unions of orbits, possibly infinitely many.
It follows that the relation $\horiztrans{}{}{} \ \subseteq Q^2$ is also equivariant, i.e., a union of orbits of $Q^2$.
Call an orbit $O \subseteq Q^2$ \emph{inhabited} if $\horiztrans{q}{}{q'}$ for some $(q, q')\in O$. 
If this is the case, since $\horiztrans{}{}{}$ is equivariant, and thus a union of orbits,
then \emph{every} pair $(q, q') \in O$ satisfies
$\horiztrans{q}{}{q'}$.
It thus makes sense to think of $\horiztrans{}{}{}$ as containing whole orbits rather than individual
elements.
Let \emph{initial-final} orbits in $Q^2$ be the ones containing pairs $(i, f)$ for $i$ initial and $f$ final state;
these orbits are determined by the unique initial and final control locations.
%

\begin{proposition} \label{claim:p2}
$L(\aaut)$ is non-empty iff an initial-final orbit in $Q^2$ is inhabited.
\end{proposition}

Likewise, the relation $\pushpop \subseteq Q^4$ is equivariant, i.e., a union of possibly infinitely many orbits in $Q^4$. 
Our aim now is to `orbitize' the rules of Lemma~\ref{lem:horiz least} so that they speak of \emph{orbits} of
pairs of states, instead of individual pairs of states, without losing any precision.

The (base) rules orbitizes easily; it speaks of \emph{diagonal} orbits, i.e., orbits of diagonal pairs $(q, q) \in Q^2$.
For treating the other rules, we need to speak of projections of $n$-tuples $w$ onto two coordinates.
We use the notation $w_{i j}$ to denote the projection of $w$ onto coordinates $i, j$, for $1\leq i < j \leq n$;
the same notation will be used for the projection of a set of tuples. 
For $O$ an orbit in $Q^n$, $O_{i j}$ is necessarily an orbit in $Q^2$.


\ignore{
We orbitize the rules of Lemma~\ref{lem:horiz least} as follows. 
\begin{itemize}
\item \emph{(orbit base)} 
every diagonal orbit in $Q^2$ is inhabited. 
\item \emph{(orbit transitivity)}
for every orbit $O_{123}$ in $Q^3$, if $O_{1 2}$ and $O_{2 3}$ are inhabited, then $O_{1 3}$ is inhabited too.
\item \emph{(orbit push-pop)}
for every orbit $O_{1234}$ in $Q^4$, 
if $O_{2 3}$ is inhabited
and $\pushpop(O_{1234})$ holds, then $O_{1 4}$ is inhabited. 
\end{itemize}
}

\begin{lemma} \label{lem:orbit least}
	An orbit $O$ of $Q^2$ is inhabited if, and only if, $\inhabited O$ is derivable according to the rules below:
	\begin{align*}
		&\textrm{\em (orbit base)} &&
			\frac{}{\inhabited O}
				&& \forall \textrm{ diag. orbit } O \textrm{ in } Q^2 \\[2ex]
		&\textrm{\em (orbit transitivity)} &&
			\frac{\inhabited {O_{12}} \quad \inhabited {O_{23}}}{\inhabited {O_{13}}}
				&& \forall \textrm{ orbit } O \textrm{ in } Q^3 \\[2ex]
		&\textrm{\em (orbit push-pop)} &&
			\frac{\inhabited {O_{23}}}{\inhabited {O_{14}}}
				&& \forall \textrm{ orbit } O \textrm{ in } \pushpop
	\end{align*}
\end{lemma}

\begin{proof}
	Both directions are proved by induction on the size of derivations.
	The ``if'' direction uses equivariance of $\horiztrans{}{}{}$.
	\ignore{
	Let $O$ be the orbit of a pair $(q, q') \in Q^2$.
	If $O$ is inhabited, by definition $\horiztrans{q}{}{q'}$ is derivable by the rules of Lemma~\ref{lem:horiz least}.
	By induction on the size of a derivation of $\horiztrans{q}{}{q'}$
	we can prove that $\inhabited O$ is derivable by the rules above.
	We prove the opposite implication by induction on the size of the derivation.
	The interesting case is transitivity.
	Fix an orbit $O_{123}$ of $Q^3$ and suppose $O_{1 3}$ is inhabited, by application of the 
	(orbit transitivity) rule to $O_{1 2}$ and $O_{2 3}$.  
	We need to show that $\horiztrans{q_1}{}{q_3}$ for all $(q_1, q_3) \in O_{1 3}$.
	%
	Choose some $q_2$ with $(q_1, q_2, q_3) \in O_{123}$.
	By induction assumption $\horiztrans{q_1}{}{q_2}$ 
	and $\horiztrans{q_2}{}{q_3}$, 
	and hence  $\horiztrans{q_1}{}{q_3}$ is derivable as required.
	}
\end{proof}

\mysubsubsection{Discretization}
The set $Q^2$ is orbit-infinite.
We encode it as a Cartesian product of an orbit-finite set and the integers $\Z$.
This will allow us to reduce non-emptiness of $\lang \Aa$ to non-emptiness of a system of equations.

Consider two states $q = \elt{l, v}, q' = \elt{l', v'} \in Q$,
where $v \in \Q^{n_l}$ and $v' \in \Q^{n_{l'}}$.
Since $Q$ is orbit finite, by Lemma~\ref{lem:bounded-of} we know that both $v$ and $v'$ have uniformly bounded span,
say $u$.
However, the joint vector $(v, v') \in \Q^{n_l + n_{l'}}$ needs not have uniformly bounded span 
(and indeed $Q^2$ is orbit-infinite),
since rationals in $v$ might be arbitrarily far from rationals in $v'$.
The idea is to ``factorize'' out the orbit infiniteness of $Q^2$
by shifting the second vector $v'$ closer to $v$ (in order to have span at most $u+1$),
and by keeping track separately of the shift as the only unbounded component.

The first technical step is to extend the tuple $v$ in every state $q = \elt{l, v} \in Q$
with one rational number $t$, written $q \concat t = \elt{l, (v, t)}$,
called the \emph{reference point} of $q \concat t$.
Reference points allow to precisely shift vectors so they become closer.
%
Let $\min v$ be the component of $v$ with minimal value.
We define
\begin{align*} 
	\dot Q = &\setof{q \concat t}{q = \elt{l, v} \in Q, t \in \Q, \min{v} \leq t < \min{v} + 1}. 
\end{align*}
The set of extended tuples $\dot Q$ is definable and orbit-finite (of uniform span at most $u+1$),
and contains exponentially many orbits.
%
%
While ${\dot Q}^2$ is not orbit-finite itself,
we can now define its subset $\psQ$ of pairs with equal reference points:
\begin{align*}
	\psQ = \setof{ (q\concat t, q'\concat t) \in {\dot Q}^2}{ t \in \Q }.
\end{align*}
%
%
Thus, $\psQ$ contains only those pairs of vectors which are close in a precise sense.
Applying Corollary~\ref{cor:nf} to $\psQ$ we obtain:
\begin{proposition} \label{prop:of}
	$\psQ$ is orbit-finite with uniform span $u+1$ and its decomposition into
	orbits is computable in \exptime.
\end{proposition} 

\newcommand{\proj}{\pi}

The idea now is to represent an arbitrary pair in $Q^2$ as an element from $\psQ$ plus an integer representing ``shift'' of the second vector.
Formally, we define the following \emph{shift mapping} $\proj : \psQ \times \Z \to Q^2$:
\[
\proj \ : \ (q \concat t, q' \concat t), z \ \mapsto \ q, q' + z,
\]
where $q'+z$ is the state obtained from $q' = \elt{l', v'}$ by adding $z$ to all time values in $v'$.
Thus the shift mapping forgets about the equal reference points of $q$ and $q'$, and shifts $q'$ by $z$.
Note that every pair of states in $Q^2$ is of the form $(q, q'+z)$, for some $z \in \Z$ and $(q \concat t, q' \concat t) \in \psQ$,
i.e., the shift mapping is surjective.
To distinguish between orbits of $Q^2$ and $\psQ$, we use lowercase $o$ for the latter.
Every orbit $O$ of $Q^2$ is the image, under the shift mapping $\proj$, of $o \times \{ z\}$, for some $z \in \Z$
and some orbit $o$ of $\psQ$.
We will call $O$ the \emph{image orbit} of $(o, z)$.
By the \emph{inverse image} of an equivariant set $X \subseteq Q^2$ we mean the set of all 
pairs $(o, z)$ whose image orbit $O$ is included in $X$.
We will call $(o, z)$ \emph{inhabited} if its image orbit is so.

The inverse image of an orbit $O$ may contain many pairs $(o, z)$, as shown in the example below,
but finitely many due to the simplifying assumption that all states are timed.
\begin{example}
	Consider the orbit $O \subseteq Q^2$ defined by
	$\phi(x, y, x') \ \equiv \ x < y < x+1 \ \land \ y+6 < x' < x+7$,
	with $x, y$ timed registers of one state and $x'$ timed register of the other.
	The inverse image of $O$ contains the pair $(o, 6)$, where $o \subseteq \psQ$ is defined by
	$x < y < t = t' = x' < x+1$,
	but also the pair $(o', 7)$, where $o' \subseteq \psQ$ is defined by
	$y-1 < x' < x = t = t' < y$.
\end{example}
%
%
The inverse image of a definable set admits a decomposition into finitely many sets of a particularly simple form:
\begin{restatable}[Decomposition Lemma]{lemma}{DecompositionLemma}
	\label{lem:decomp}
	For a definable subset $X \subseteq Q^2$, its inverse image decomposes into a finite union
	of sets of the form
	\[
		\{o\} \times I,
	\]
	where $o$ is an orbit in $\psQ$, and $I \subseteq \Z$
	is one of 
\begin{align*}
	\Z_{<m} = \{ z  : z < m \}, \quad  \{ m \}, \quad 
	\Z_{>m} = \{ z  : z > m \},   
\end{align*}
for $m \in \Z$.
	A decomposition of $X$ is computable in \exptime.  
\end{restatable}
\noindent
The following corollary will be useful later:
\begin{proposition} \label{prop:invim}
The inverse image of an orbit $O \subseteq Q^2$ is finite and computable in \exptime.
\end{proposition}

We are going to define a system of equations $\Delta$,
with variables $X_o$ corresponding to orbits $o$ in $\psQ$. 
The construction will conform to the following correctness condition:
\begin{lemma} \label{lem:correctness}
	The least solution $\nu$ of the system $\Delta$ assigns to a variable $X_o$ the set
	\[
	\nu(X_o) \ = \ \{ z \in \Z \ : \ (o, z) \text{ is inhabited} \}.
	\]
\end{lemma}
\noindent
Orbits $o \subseteq \psQ$ that appear in the inverse image of an initial-final orbit $O \subseteq Q^2$
we call initial-final too; again, they are determined by the unique initial and final control locations.
Based on the last lemma, we reformulate Proposition~\ref{claim:p2} as:
%
\begin{proposition} \label{claim:p3}
	$L(\aaut)$ is non-empty iff $\nu(X_o)$ is non-empty, for an initial-final orbit $o$.
\end{proposition}
Thus non-emptiness of $L(\aaut)$ reduces in \exptime to non-emptiness of some of the variables $X_o$ in $\Delta$.
%

To complete the proofs of upper bounds of Theorems~\ref{thm:nexptime} and~\ref{thm:exptime},
we need to describe the construction of $\Delta$ and prove that it verifies the condition in Lemma~\ref{lem:correctness}.

\mysubsubsection{System of equations}  \label{subsec:system:equations}
When defining $\Delta$ we prefer to use inclusions.
Roughly speaking, the system $\Delta$ corresponds to the inverse image of the rules in Lemma~\ref{lem:orbit least}.

Consider the (orbit base) rule first. We observe that all orbits $o$ appearing in the inverse image of a diagonal orbit $O$ 
are diagonal as well.
Thus for every diagonal orbit $o$ in $\psQ$ we add to $\Delta$ the inclusion
\begin{align}
	X_{o} \ \supseteq \ \{0\}.
\end{align}
%

For treating the (orbit transitivity) rule we need to extend the shift mapping $\proj$ from pairs to triples.
Define the set of triples of states with equal reference points
\[
\pspsQ \ = \ \{ (q\concat t, q'\concat t, q'' \concat t) \in {\dot Q}^3 \ : \ t \in \Q \},
\]
and consider the shift mapping $\proj : \pspsQ \times \Z^2 \to Q^3$:
\[
(q \concat t, q' \concat t, q'' \concat t), z, z' \ \mapsto \ q, q' + z, q'' + z + z'.
\]
As before, $\pi$  transforms a triple $(o, z, z')$, where $o$ is an orbit in $\pspsQ$, into an orbit $O$ in $Q^3$.
For an orbit $O$ in $Q^3$, consider any element $(o, z, z')$ of its inverse image, i.e., $O$ is the image of $(o, z, z')$.
The image commutes with projections, i.e., $O_{1 2}$ is necessarily the image of $(o_{1 2}, z)$, and likewise
$O_{2 3}$ and $O_{1 3}$ are images of $(o_{2 3}, z')$ and $(o_{1 3}, z+z')$, respectively.
Therefore the (orbit transitivity) rule says that if $(o_{1 2}, z)$ and $(o_{2 3}, z')$ are inhabited,
then $(o_{1 3}, z + z')$ is inhabited too.
Thus, for every orbit $o$ in $\pspsQ$ we add the following inclusion to $\Delta$:
\begin{align}
	X_{o_{1 3}} \ \supseteq \ X_{o_{1 2}} + X_{o_{2 3}} .
\end{align}

Finally, we address the (orbit push-pop) rule. 
We consider separately two cases, depending on whether the stack symbol pushed/popped is timeless or timed.
Each of the two cases will induce separate inclusions in $\Delta$.
Let $S$ be partitioned into timeless stack symbols $S_0$ and timed stack symbols $S_1$.
$S_0$ is a finite set.
We partition $\pushpop$ into $\pushpopA$ and $\pushpopB$, where
\begin{align*}
	\pushpopA &= \setof
		{(q, \bar q, \bar q', q')}
		{\exists \bar s \in S_0. \begin{array}{c} \rhopush(q, \bar q, \bar s), \\ \rhopop(\bar q', \bar s, q')\end{array}} \\
	\pushpopB &= \setof
		{(q, \bar q, \bar q', q')}
		{\exists \bar s \in S_1. \begin{array}{c} \rhopush(q, \bar q, \bar s), \\ \rhopop(\bar q', \bar s, q')\end{array}}
\end{align*}
\ignore{
Since $\pushpop$ speaks about quadruples,
we extend the shift mapping $\proj$ to quadruples, in the natural way.
Define the set of quadruples of states with equal reference points
\[
\pspspsQ \ = \ \{ (q\concat t, \bar q\concat t, \bar q'\concat t, q' \concat t) \in {\dot Q}^4 \ : \ t \in \Q \},
\]
and consider the following shift mapping $\proj$ from $\pspspsQ \times \Z^3$ to $Q^4$:
\[
(q \concat t, \bar q \concat t, \bar q' \concat t, q' \concat t), z, \bar z, z' \ \mapsto \ 
q, \bar q + z, \bar q' + z + \bar z, q' + z + \bar z + z'.
\]
As before, $\pi$ transforms an orbit $o$ in $\pspspsQ$ into an orbit $O$ in $Q^4$.
For an orbit $O$ in $Q^4$ satisfying $\pushpop(O)$  (cf.~\eqref{eq:rule}),
consider any element $(o, z, \bar z, z')$ of its inverse image, i.e., $O$ is the image of $(o, z, \bar z, z')$.
As before, the image commutes with shift mapping,
i.e., $O_{1 2}$, $O_{2 3}$, $O_{3 4}$ and $O_{1 4}$ are necessarily images of 
 $(o_{1 2}, z)$, $(o_{2 3}, \bar z)$ and $(o_{3 4}, z')$ and $(o_{1 4}, z + \bar z + z')$, respectively.
Therefore the (orbit push-pop) rule says that if $(o_{2 3}, \bar z)$ is inhabited,
then $(o_{1 4}, z + \bar z + z')$ is inhabited too.
}

First, we consider the (orbit push-pop) rule restricted to only timeless stack symbols. 
We can write $\pushpopA$ as a finite sum of products
\begin{align*}
	\pushpopA & = \bigcup_{\bar s \in S_0} \rhopush_{\bar s} \times \rhopop_{\bar s},
\end{align*}
where $\rhopush_{\bar s}(q, \bar q) \equiv \rhopush(q, \bar q, \bar s)$
and $\rhopop_{\bar s}(\bar q', q') \equiv \rhopop(\bar q', \bar s, q')$.
%
For a fixed $\bar s \in S_0$, $\rhopush_{\bar s}$ and $\rhopop_{\bar s}$ are definable subsets of $Q^2$,
and thus Lemma~\ref{lem:decomp} applies.
%
%
%
%

We need to extend once more the shift mapping $\proj$, this time to quadruples.
Define the set of quadruples of states with equal reference points
\[
\pspspsQ \ = \ \{ (q\concat t, \bar q \concat t, \bar q' \concat t, q' \concat t) \in {\dot Q}^4 \ : \ t \in \Q \},
\]
and consider the shift mapping $\proj$ from $\pspspsQ \times \Z^3$ to $Q^4$:
\[
(q \concat t, \bar q \concat t, \bar q' \concat t, q' \concat t), z, \bar z, z' \ \mapsto \ q, \bar q + z, \bar q' + z + \bar z, 
q' + z + \bar z + z'.
\]
Similarly as before, $\pi$  transforms a quadruple $(o, z, \bar z, z')$, where $o$ is an orbit
in $\pspspsQ$, into an orbit $O$ in $Q^4$.
Similarly as before we define the inverse image of $O \subseteq Q^4$.

The (orbit push-pop) rule says that if $(o_{2 3}, \bar z)$ is inhabited,
$(o_{1 2}, z)$ belongs to the inverse image of $\rhopush_{\bar s}$ and
$(o_{3 4}, z')$ belongs to the inverse image of $\rhopop_{\bar s}$, then
$(o_{1 4}, z + \bar z + z')$ is inhabited.
Therefore for every orbit $o \subseteq \pspspsQ$ appearing in the inverse image of $\pushpop$,
for every $\bar s \in S_0$,
for every pair of intervals $I, I'$ such that $(o_{1 2}, I)$ appears in the decomposition of $\rhopush_{\bar s}$
and $(o_{3 4}, I')$ appears in the decomposition of $\rhopop_{\bar s}$ (by Lemma~\ref{lem:decomp}), 
we add to $\Delta$ the inclusion
\begin{align}
	X_{o_{1 4}} \ \supseteq \ X_{o_{2 3}} + \ Z_{I + I'},
\end{align}
where $Z_{I + I'}$ is a variable that, in the least solution, is assigned the set of integers $I + I'$ 
(cf. Example~\ref{ex:equations}).
This completes the proof of the upper bound of Theorem~\ref{thm:exptime}.

In order to complete the proof of Theorem~\ref{thm:nexptime}, we consider now 
the (orbit push-pop) rule restricted to only \timed stack symbols.
For convenience, we extend $\pushpopB$ with the stack symbol and consider
%
\begin{align*}
	\pushpopC &= \setof
		{(q, \bar q, \bar q', q', \bar s) \in Q^4 \times S_1}
		{\begin{array}{c} \rhopush(q, \bar q, \bar s), \\ \rhopop(\bar q', \bar s, q')\
		\end{array}}
\end{align*}
Since we are considering orbit-finite \trPDA, $\rhopush_{23}$ and $\rhopop_{12}$ are orbit-finite.
Thus, $\pushpopC_{235}$ is orbit-finite as well (in passing we extend the notation for projection 
from pairs to to triples of coordinates), due to the restriction to timed stack symbols only.
Indeed, the uniform bound on the span of $\pushpopC_{235}$ is at most twice as large as the universal bound on span of sets
$\rhopush_{23}$ and $\rhopop_{12}$.
By Corollary~\ref{cor:nf} we may enumerate all orbits $O \subseteq \pushpopC_{235}$ in \exptime.

Consider every orbit $O \subseteq \pushpopC_{235}$ separately.
We transform the set $\pushpopC$
into normal form (using Lemma~\ref{lem:nf}) and apply Projection Lemma~\ref{lem:proj} to deduce that the set
\begin{align*}
	X_O &=	\setof{(q, q') \in Q^2}
			{\exists(\bar q, \bar q', \bar s) \in O \cdot
				\begin{array}{c} \rhopush(q, \bar q, \bar s), \\ \rhopop(\bar q', \bar s, q')\end{array}}
\end{align*}
is definable and computable in \exptime.
For every $(\bar o, \bar z)$ in the inverse image of $O_{12} \subseteq Q^2$
(we use Proposition~\ref{prop:invim} here),
and for every $(o, I)$ in the decomposition of $X_O$ (by Decomposition Lemma~\ref{lem:decomp}),
we add to $\Delta$ the inclusion
\begin{align}
	X_{o} \ \supseteq \ \ (X_{\bar o} \cap \{ \bar z\}) \ + \ Z_{I -\bar z},
\end{align}
where $I - \bar z = \{ z - \bar z \ : \ z \in I\}$.
This completes the construction of $\Delta$.
Since $\Delta$ is of exponential size,
we can solve it NEXPTIME according to Lemma~\ref{lem:setsofeq:NP}.
This concludes the proof of Theorem~\ref{thm:nexptime}.



\section{Conclusions and future work}

We have investigated the reinterpretation of the classical definition of pushdown automata 
in the setting of sets with timed atoms, called \trPDA.
In order to relate to the previous research we identified the subclass of \trPDA with \emph{timeless stack},
and shown that dense-timed PDA of~\cite{AbdullaAtigStenman:DensePDA:12} 
can be effectively transformed into this subclass.

The rest of the paper focused on the non-emptiness analysis of \trPDA.
We showed that the non-emptiness problem for unrestricted \trPDA is undecidable, 
but decidable in \nexptime for orbit-finite \trPDA. 
Furthermore, non-emptiness for an even smaller subclass of  \trPDA with timeless stack has been shown \exptime-complete. 
The last result subsumes the \exptime-completness of \dtPDA~\cite{AbdullaAtigStenman:DensePDA:12},
by our language-preserving transformation of \dtPDA to \trPDA with timeless stack.

As future research, it remains to be closed the complexity gap for orbit-finite \trPDA,
as well as the detailed study of expressive power of different subclasses of \trPDA.
Moreover, in this paper we did not consider all reasonable subclasses of \trPDA.
For instance, we do not know the decidability status of non-emptiness of \emph{lhs orbit-finite} \trPDA, defined like
orbit-finite \trPDA but with the orbit-finiteness restriction imposed on the left-hand sides of transition rules only.
With respect to non-emptiness, the class is equivalent to the superclass 
of short form \trPDA (cf.~Sec.~\ref{sec:trPDA}), obtained by dropping the orbit-finiteness restriction on the rhs of
$\rhopush$ and on the lhs of $\rhopop$.
Our reduction, when extended to this model, yields systems of equations over sets of integers
that use intersections with arbitrary intervals.
Decidability of such extended systems of equations is, up to our knowledge, an open problem,
interesting on its own.
 
Finally, first-order definable sets may be considered for other atoms.
We have recently studied the reachability analysis for PDA for the important class of \emph{oligomorphic} atoms
(i.e., $\atoms^n$ is orbit-finite for every $n$) in \cite{ClementeLasota:FO-PDS:arxiv:2015},
where most of the subclasses of PDA defined in this paper become expressively equivalent.
This covers many examples, such as total order atoms $(\Q, \leq)$,
partial order atoms, tree order atoms, and many more \cite{survey}.


\bibliographystyle{IEEEtran}
\bibliography{../bib}

\newpage
\appendices


\section{Proof of Theorem~\ref{thm:Abdulla}}


\subsection{Preliminaries}

A \emph{configuration} $c$ of a \dtPDA as above is a tuple $\tuple {p, \mu, r}$,
where $p \in L$ is a control location,
$\mu : X \to \Q^{\geq 0}$ is a \emph{clock valuation} for the clocks in $X$,
and $r \in (\Gamma\times\Q^{\geq 0})^*$ is the \emph{stack valuation},
recording stack symbols and their current age.
For a clock valuation $\mu : X \to \Q^{\geq 0}$ and a number $k \in \Q^{\geq 0}$,
we denote by $\mu + k$ the valuation which adds $k$ to every clock,
i.e., for every $x \in X$, $(\mu + k)(x) = \mu(x) + k$.
Similarly, for a stack valuation $r = (\gamma_1, h_1)\dots(\gamma_n, h_n)$ and $k \in \Q^{\geq 0}$,
we write $r + k$ for the new valuation $(\gamma_1, h_1 + k)\dots(\gamma_n, h_n + k)$.
Given a clock valuation $\mu : X \to \Q^{\geq 0}$, a set of clocks $Y \subseteq X$,
and a new value $k \in \Q^{\geq 0}$,
we denote with $\mu[Y \leftarrow k]$ the new clock valuation which is the same as $\mu$ except that assigns value $k$ to all clocks in $Y$,
i.e., $\mu[Y \leftarrow k](y) = k$ if $y \in Y$ and $\mu[Y \leftarrow k](y) = \mu(y)$ if $y \not\in Y$.
Similarly, we denote with $\mu[z \leftarrow k] : (X \cup \set z) \to \Q^{\geq 0}$
the extended clock valuation where the value of the special clock $z$ is $k$. 
Given an extended clock valuation $\mu' : (X\cup\{z\}) \to \Q^{\geq 0}$ and a formula $\varphi$,
we write $\mu' \models \varphi$ if $\varphi$ holds when clock variables are replaced by values given by $\mu'$.
As usual, we distinguish transitions representing elapse of time, which are labelled by some $t \in \Q^{\geq 0}$,
and discrete transitions, which for convenience are labelled by tuples of the form $(a, \phi, Y, op)$.
Formally, for every $t \in \Q^{\geq 0}$ we have a timed transition $\tuple {p, \mu, r} \goesto t \tuple {p, \mu + t, r + t}$,
and, if $p \goesto {a, \phi, Y, op} q$,
then we have a discrete transition $\tuple {p, \mu, r} \goesto {a, \phi, Y, op} \tuple {q, \mu', r'}$,
whenever $\mu \models \varphi$ holds,
$\mu' = \mu[Y \leftarrow 0]$, and, depending on the kind of operation $op$,
\begin{itemize}
	\item Case $op = \nop$: $r' = r$.
	\item Case $op = \push{\alpha \models \psi}$: $r' = r(\alpha, k)$ if $\mu[z \leftarrow k] \models \psi$.
	\item Case $op = \pop{\alpha \models \psi}$: $r = r'(\alpha, k)$ if $\mu[z \leftarrow k] \models \psi$.
\end{itemize}
A run $\pi$ is an alternating sequence $c_0 tr_0 \dots c_k$ of configurations $c_i$'s and transitions $tr_i$'s
s.t., for every $0 \leq i < k$, $c_i \goesto {tr_i} c_{i+1}$.

\subsection{Simplifications}

The following simplifying assumptions can be shown not to reduce the recognizing power of the model.
They are either standard, or very easy to show.

\paragraph{Only pop constraints of the form $z - x \sim k$}

Constraints of the form $x - z \sim k$ can be converted to the form $z - x \sim k$
by negating both sides and by flipping the inequality.

\paragraph{No pop constraints of the form $z \sim k$}

We introduce a new clock $x_0$ which is ensured to be $0$ when a pop transition is taken.
A constraint $z \sim k$ can thus be replaced by $z - x_0 \sim k$.
Formally, a pop transition $p \goesto {a, \phi, T, \pop{\alpha \models \psi}} q$ is simulated in two steps:
\begin{align*}
	p &\goesto {a, \phi, \{x_0\}, \nop}  (p, \alpha, \psi', q) &(1) \\
	(p, \alpha, \psi', q) &\goesto {\varepsilon, \phi \wedge x_0 = 0, T, \pop{\alpha \models \psi'}} q &(2)
\end{align*}
where $\psi'$ is the same as $\psi$ where all constraints of the form $z \sim k$ are replaced by $z - x_0 \sim k$.
Transition (1) mimics exactly the original transition, except that the pop operation is not performed.
Transition (2) is ensured to be taken with delay $0$ after (1),
and the pop operation is thus performed under the condition that $x_0$ is $0$.



\ignore{
\paragraph{Normalized pop constraints}

A conjunctive constraint $\psi$ is \emph{normalized} if,
for every clock $x$, it contains at most one expression of the form $z - x \lesssim k$,
and similarly for $z - x \gtrsim k$.
%
Every conjunctive constraint $\psi$ can be normalized by removing redundant constraints with no increase in complexity.
We assume that pop constraints are conjunctive and normalized.
}

\paragraph{No resets on push and pop transitions}

We wish that clocks are reset only on nop transitions.
To achieve this, we introduce an extra clock $x^*$ and some auxiliary intermediate control states.
A push/pop transition of the form
$p \goesto {a, \phi, T, op}  q$
is split into three consecutive transitions
\begin{align*}
	p &\goesto {a, \phi, \{ x^* \}, \nop} (p, a, \phi, T, op, 0) &(1) \\
	(p, a, \phi, T, op, 0) &\goesto {\varepsilon, x^* = 0, \emptyset, op} (p, a, \phi, T, op, 1) &(2)\\
	(p, a, \phi, T, op, 1) &\goesto {\varepsilon, x^* = 0, T, \nop} q &(3)
\end{align*}
Transition (1) reads the input, checks the clock constraint $\phi$, but it does not reset clocks $T$, neither perform the stack operation.
Transition (2) performs the actual stack operation (without resetting any clock),
and the constraint on $x^*$ ensures that no time elapsed since (1).
Finally, transition (3) resets the clocks in $T$, and again no time is allowed to elapse.
Clocks in $T$ cannot be reset directly by transition (1)
since in the semantics of push/pop operations we must compare the age of the stack symbol to the value of clocks prior to their reset.

\paragraph{The initial age is $0$}

A push operation $\push{\alpha \models \psi_1}$ can be restricted to have the push constraint $\psi_1$ of the trivial form $z = 0$,
i.e., the initial age is always $0$.
We omit a trivial push constraint by just writing $\push{\alpha}$.

A conjunctive constraint $\psi_1$ involving only clock $z$
is equivalent to either a punctual constraint $z = k$ for an integer $k$,
or to an interval constraint $z \in (a, b)$
for a lower bound $a \in \Z \cup \{-\infty\}$
and an upper bound $b \in \Z \cup \{+\infty\}$.
The idea is to push this information on the stack,
which is used at the time of pop to update the pop constraint.

The function $update(\psi, I)$ for $I$ equal to either $k$ or $(a, b)$ as above,
is defined by structural induction on $\psi$ and it works by shifting all constraints by the amount specified by $I$:
\begin{align*}
	update(z - x \sim h, k) &= z - x \sim h - k \textrm{ for } \sim \in \{ \leq, <, >, \geq \} \\
	update(z - x \lesssim h, (a, b)) &= z - x < h - a \textrm{ for } \lesssim \in \{ \leq, < \} \\
	update(z - x \gtrsim h, (a, b)) &= z - x > h - b \textrm{ for } \gtrsim \in \{ \geq, > \} \\
	update(\true, I) &= \true \\
	update(\psi_0 \wedge \psi_1, I) &= update(\psi_0) \wedge update(\psi_1)
\end{align*}
A push transition $p \goesto {a, \phi, Y, \push{\alpha \models \psi}} q$
becomes $p \goesto {a, \phi, Y, \push{\tuple{\alpha, I}}} q$
where $I$ is either $k$ or $(a, b)$ as implied by the constraint $\psi$,
and a pop transition $p \goesto {a, \phi, Y, \pop{\alpha \models \psi}} q$
is replaced by several transitions of the form $p \goesto {a, \phi, Y, \pop{\tuple{\alpha, I} \models \psi'}} q$,
for every shift $I$,
and where $\psi' = update(\psi, I)$ is the pop constraint updated by $I$.
Correctness for a punctual constraint is immediate.
For $z \in (a, b)$ and a pop constraint $z - x \lesssim h$,
the semantics asks whether there exists an initial age $z_0 \in (a, b)$ s.t. $z - x \lesssim h$ holds at the time of pop,
which is the same as requiring $(z + z_0) - x \lesssim h$ when the initial age is $0$ instead,
that is, $z - x \lesssim h - z_0$.
This constraint is easier to satisfy for smaller values of $z_0 > a$,
and thus we obtain $z - x < h - a$.
The reasoning for $z - x \gtrsim h$ is analogous.

\subsection{Untiming the stack}

We are now ready present the formal construction for untiming the stack of a \dtPDA.
Let $\taut = \tuple {Q, q_0, \Sigma, \Gamma, X, z, \Delta}$ be a \dtPDA.
We construct a \dtPDA with timeless stack $\uaut = \tuple {Q', q_0', \Sigma, \Gamma', X', \Delta'}$.
Let
\begin{align*}
	\mathcal C_\sim = \setof {(x, \sim, k)} {z - x \sim k \in \psi \textrm{ for a pop constraint } \psi}
\end{align*}
The new set of clocks $X'$ is obtained by adding to $X$ a clock $\hat x_{\sim k}$
for every $(x, \sim, k) \in \mathcal C$,
where $\mathcal C = \bigcup_\sim\mathcal C_\sim$.
A control state in $\uaut$ is of the form $(p, R, O)$,
where $p$ is a control state in $\taut$.
The set $R \subseteq \mathcal C_< \cup \mathcal C_\leq$ represents active reset restrictions,
and the set $O \subseteq \mathcal C_> \cup \mathcal C_\geq$ represents active reset obligations.
The initial control state of $\uaut$ is $q_0' = (q_0, \emptyset, \emptyset)$.
The stack alphabet of $\uaut$ consists of tuples $\tuple{\alpha, \psi, R, O}$,
where $\alpha \in \Gamma$ is a stack symbol of $\taut$,
$\psi$ is a clock constraint, and $R, O$ are as above.

Transitions in $\uaut$ are defined as follows.
If we have a push transition $p \goesto {a, \phi, \emptyset, \push{\alpha}} q$ in $\taut$,
then we have several push transitions in $\uaut$ of the form
\begin{align*}
	(p, R, O) \goesto {a, \phi', T', \push{\alpha'}} (q, R', O_1') & \\
	\textrm{ with } \alpha' = \tuple{\alpha, \psi, R, O_0'}
\end{align*}
for every $R, R' \subseteq \mathcal C_< \cup \mathcal C_\leq$,
$O, O_0', O_1' \subseteq \mathcal C_> \cup \mathcal C_\geq$,
and constraints $\psi, \phi'$ satisfying the conditions below.
Constraint $\psi$ is guessed to be the constraint that will hold at the time of the corresponding pop.
%
%
The other components are determined as follows.
We first consider reset restrictions.
Let $M = \setof{ (x, \lesssim, k) } {z - x \lesssim k \in \psi}$
be the new reset restrictions as implied by the guessed pop constraint $\psi$.
Since restrictions in $R$ subsume those in $M$,
we reset $\hat x_{\lesssim k}$ only for restrictions in $M \setminus R$.
\begin{align*}
	R'		&= R \cup M \\
	T_0 	&= \setof {\hat x_{\lesssim k}}{ (x, \lesssim, k) \in M \setminus R } \\
	\phi_0	&= \phi \wedge \bigwedge_{ (x, \lesssim, k) \in M \setminus R } - x \lesssim k
\end{align*}
We now address reset obligations.
Let $N = \setof{ (x, \gtrsim, k) } {z - x \gtrsim k \in \psi}$
be the new reset obligations.
New reset obligations in $N$ always subsume previous ones in $O$.
%
Let $O'_0 \subseteq O \setminus N$ be any set of previous reset obligations not subsumed by new ones.
Intuitively, we guess that obligations in $O'_0$ will be satisfied after the matching pop,
thus we push them on the stack.
Let $N' \subseteq N$ be those new reset obligations which are already satisfied by a previous reset,
and thus need not be added.
\begin{align*}
	O_1'		&= (O \setminus O'_0) \cup (N \setminus N') \\
	T'			&= T_0 \cup \setof { \hat x_{\gtrsim k} } { (x, \gtrsim, k) \in N \setminus N' } \\
	\phi'		&= \phi_0 \wedge \bigwedge_{ (x, \gtrsim, k) \in N'} - x \gtrsim k 
\end{align*}

For every pop transition $p \goesto {a, \phi, \emptyset, \pop{\alpha \models \psi}} q$ in $\taut$,
we have an ``untimed'' pop transition in $\uaut$ of the form
\begin{align*}
	(p, R, \emptyset) \goesto {a, \phi, \emptyset, \pop{\alpha'}} (q, R', O') & \\
	\textrm{ with } \alpha' = \tuple{\alpha, \psi, R', O'}
\end{align*}
for every $R, R' \subseteq \mathcal C_< \cup \mathcal C_\leq$ and $O' \subseteq \mathcal C_> \cup \mathcal C_\geq$.
We require the set of reset obligations to be empty in order to ensure
that all clocks that were under a reset obligation have been indeed reset.
%
%

For every nop transition $p \goesto {a, \phi, T, \nop} q$ in $\taut$,
we have a nop transition in $\uaut$ of the form
\begin{align*}
	(p, R, O) \goesto {a, \phi', T, \nop} (q, R, O \setminus O')
\end{align*}
for every $R \subseteq \mathcal C_< \cup \mathcal C_\leq$,
 $O, O' \subseteq \mathcal C_> \cup \mathcal C_\geq$,
and for every set of reset obligations $O' \subseteq \setof {(x, \gtrsim, k) \in O} {x \in T}$
which are satisfied by a reset in this transition:
\begin{align*}
	\phi' = \phi \wedge
	\!\!\!\! \bigwedge_{x \in T, (x, \lesssim, k) \in R} \!\!\!\! \left( \hat x_{\lesssim k} \lesssim k \right)
	\wedge
	\!\!\!\! \bigwedge_{ (x, \gtrsim, k) \in O' } \!\!\!\! \left( \hat x_{\gtrsim k} \gtrsim k \right)
\end{align*}

\thmDTPDAuntimedstack*

\begin{proof}
	
	Let $\pi$ be an accepting run in $\taut$,
	\begin{align*}
		\pi = (p_0, \nu_0, v_0) tr_0 (p_1, \nu_1, v_1) \cdots (p_{k+1}, \nu_{k+1}, v_{k+1}) & \\
		\textrm{ with } tr_i = (a_i, \phi_i, T_i, op_i)
	\end{align*}
	We construct an accepting run $\pi'$ in $\uaut$,
	\begin{align*}
		\pi' = (r_0, \mu_0, u_0) tr_0' (r_1, \mu_1, u_1) \cdots (r_{k+1}, \mu_{k+1}, u_{k+1}),& \\
		\textrm{ with } r_i = (p_i, R_i, O_i) & \\
		\textrm{ and } tr_i' = (a_i, \phi_i', T_i', op_i')
	\end{align*}
	where $tr_i' = tr_i$ if $tr_i \in \R$, and otherwise it is determined as follows.
	For every $i \leq j$, let $t_{i, j}$ be the total time elapsed from transition $i$ to transition $j$, i.e.,
	\begin{align*}
		t_{i, j} = \sum_{h = i}^j \left\{ \begin{array}{ll} 
			tr_h	& \textrm{ if } tr_h \in \R \\
			0		& \textrm{ otherwise }
		\end{array}\right.
	\end{align*}
	If $i > j$, we define $t_{i, j} = -t_{j, i}$.
	\ignore{
	%
	Let $h_i$ be the height of the stack, and let $O^1_i, \dots, O^{n_i}_i$ be the contents of the stack projected to the fourth component.
	Conventionally, let $O^0_i = O_i$.
	We assume w.l.o.g. that the first transition is a push and that the last one is the matching pop,
	thus the stack is empty only at the beginning of the run and at the end.
	\lorenzo{assume that all clocks are reset on the first transition}
	For every $1 \leq j \leq h_i$, let $m_i^j$ be the index of the transition when the $j$-th topmost symbol of the stack will be popped,
	and let $l_i^j$ be the corresponding push.
	We have
	\begin{align*}
		0 = l_i^{h_i} < \cdots < l_i^1 \leq l_i^0 = i = m_i^0 \leq m_i^1 < \cdots < m_i^{h_i} = k
	\end{align*}
	where we take $l_i^0 = i = m_i^0$.
	We maintain the following invariant: $(x, \gtrsim, k) \in O_i^j$ if, and only if,
	\begin{align}
		\label{eq:invariant}
		z - x \gtrsim k \in \psi_{m_i^{j+1}} \wedge
		\exists (m_i^j < h < m_i^{j+1}) \cdot
		\left\{\begin{array}{c}
			x \in T_h \\ \wedge \\
			t_{l_i^j, h} \gtrsim k
		\end{array}\right.
	\end{align}
	That is, reset obligations in $O_i^j$ will be satisfied by a reset on $x$
	more than $k$ time units since the corresponding push.
	In particular, $O_{j_i} = \empty$.
	}
	The construction of $\pi'$ is based on the following two observations.
	\begin{enumerate}
		
		\item For any reset restriction $z - x \lesssim k \in \psi_{j_i}$,
		whenever $x \in T_h$ is reset at transition $tr_h$ with $h < j_i$,
		the time elapsed between $tr_i$ and $tr_h$ is $t_{i, h} \lesssim k$.
		
		\item For any reset obligation $z - x \gtrsim k \in \psi_{j_i}$,
		there exists a \emph{minimal} index $h < j_i$ s.t. 
		$x \in T_{h}$ is reset at transition $tr_{h}$
		and $t_{i, h} \gtrsim k$.
		(Minimality is important to construct a run in $\uaut$,
		in order to mimic the fact that new reset obligations truly subsume old ones.)
		
	\end{enumerate}
	We proceed by a case analysis on $op_i$.
	Let $op_i = \push{\alpha_i}$.
	The corresponding pop operation has $op_{j_i} = \pop{\alpha_{j_i} \models \psi_{j_i}}$,
	with $\alpha_{j_i} = \alpha_i$.
	By assumption, $T_i = \emptyset$.
	Take $tr_i' = (a_i, \phi_i', T_i', op_i')$
	with $op_i' = \tuple{\alpha_i, \psi_{j_i}, R_i, O^0}$,
	where $\phi_i'$, $T_i'$, and $O^0$ are defined as follows.
	We first analyse reset restrictions.
	Let $M = \setof{ (x, \lesssim, k) } {z - x \lesssim k \in \psi_{j_i}}$
	be the set of reset restrictions due to $\psi_{j_i}$, and let
	\begin{align*}
		T^0		&= \setof {\hat x_{\lesssim k}}{ (x, \lesssim, k) \in M \setminus R_i } \\
		\phi^0	&= \phi_i \wedge \bigwedge_{ (x, \lesssim, k) \in M \setminus R_i } - x \lesssim k
	\end{align*}
	We show $\mu_i \models \phi^0$.
	First, $\mu_i \models \phi_i$ holds because $\pi$ is a valid run in $\taut$.
	Let $(x, \lesssim, k) \in M \setminus R_i$.
	Then, $\mu_{j_i} \models z - x \lesssim k$.
	If $k \geq 0$, then $\mu_i \models -x \lesssim k$ immediately holds.
	If $k < 0$, let $h$ be the last transition before $tr_i$ when $x$ is reset.
	By Point 1) above, $t_{i, h} \lesssim k$,
	i.e., the last reset of $x$ is more than $-k$ time units before transition $i$.
	Thus, $\mu_i \models x \gtrsim -k$.
	
	We now analyse reset obligations.
	Let $N = \setof{ (x, \gtrsim, k) } {z - x \gtrsim k \in \psi_{j_i}}$
	be the reset obligations due to $\psi_{j_i}$,
	and let $N' = \setof { (x, \gtrsim, k) \in N } {\mu_i \models - x \gtrsim k}$
	be those obligations in $N$ which are already satisfied by a past reset.
	(Necessarily $k \leq 0$ for $(x, \gtrsim, k) \in N'$.)
	For $(x, \gtrsim, k) \in O_i$, let $l$ be the largest index $< i$
	s.t. $(x, \gtrsim, k) \in O_{l+1} \setminus O_l$.
	Then, $tr_l$ is a push transition with matching pop transition $tr_{j_l}$
	with $(x, \gtrsim, k) \in \psi_{j_l}$.
	By Point 2) above, there exists $h < j_l$ s.t. $x \in T_h$ and $t_{l, h} \gtrsim k$.
	Let $O^0 = \setof{ (x, \gtrsim, k) \in O } { h > j_i }$
	be those obligations which will be satisfied after the matching pop transition $tr_{j_i}$.
	Obligations in $O^0$ are pushed on the stack.
	Then, let
	\begin{align*}
		O_{i+1} &= (O_i \setminus O^0) \cup (N \setminus N') \\
		T_i'	&= T^0 \cup \setof { \hat x_{\gtrsim k} } { (x, \gtrsim, k) \in N \setminus N' } \\
		\phi_i'	&= \phi_0 \wedge \bigwedge_{ (x, \gtrsim, k) \in N'} - x \gtrsim k
	\end{align*}
 	Clearly, $\mu_i \models \phi_i'$ holds,
	since we proved above $\mu_i \models \phi_0$, and by the definition of $N'$.
	
	Let's now analyse the corresponding pop operation $op_{j_i} = \pop{\alpha_{j_i} \models \psi_{j_i}}$.
	Once again, $T_{j_i} = \emptyset$ by assumption.
	By construction of $\pi'$ (cf. the push transition above),
	$tr_i'$ pushed a symbol of the form $\tuple{\alpha_i, \psi_i, R_i, O^0}$,
	with $\alpha_i = \alpha_{j_i}$ and $\psi_i = \psi_{j_i}$.
	Therefore, take $tr_{j_i}' = (a_{j_i}, \phi_{j_i}, \emptyset, op_{j_i}')$
	with $op_{j_i}' = \pop{\tuple{\alpha_i, \psi_i, R_i, O^0}}$
	and define $R_{{j_i}+1} := R_i$ and $O_{{j_i}+1} := O^0$.
	By construction, reset obligations added to $O_{j_i}$ are removed as soon as they can be satisfied
	(cf. the definition of $O'$ in the nop rule below).
	All reset obligations can be satisfied by Point 2) above.
	Thus, $O_{j_i} = \emptyset$, and $tr_{j_i}'$ is a valid transition.
	
	Finally, Let $op_i = \nop$.
	Let $O'$ be defined as
	\begin{align*}
	 	O' = \setof {(x, \gtrsim, k) \in O_i} { x \in T_i \textrm{ and } \mu_i \models \hat x_{\gtrsim k} \gtrsim k}
	\end{align*}
	Take $tr_i' = (a_i, \phi'_i, T_i', \nop)$, with $T_i' = T_i$ and
	\begin{align*}
		\phi'_i &= \phi_i \wedge \phi^0 \wedge \phi^1, \textrm{ where } \\
		\phi^0	&= \bigwedge_{x \in T_i, (x, \lesssim, k) \in R_i} \hat x_{\lesssim k} \lesssim k \\
		\phi^1	&= \bigwedge_{ (x, \gtrsim, k) \in O' } \hat x_{\gtrsim k} \gtrsim k
	\end{align*}
	and let $R_{i+1} = R_i$ and $O_{i+1} = O_i \setminus O'$.
	We show that $\mu_i \models \phi'_i$.
	\begin{itemize}
		
		\item $\mu_i \models \phi_i$ since $\pi$ is a valid run in $\taut$.
		
		\item $\mu_i \models \phi^0$:
		Let $x \in T_i$ and $(x, \lesssim, k) \in R_i$.
		We show	$\mu_i \models \hat x_{\lesssim k} \lesssim k$.
		Let $h^*$ be the largest index $h < i$ s.t. $(x, \lesssim, k) \in R_{h+1} \setminus R_h$.
		Then, $tr_{h^*}$ is a push transition, and $\hat x_{\lesssim k} \in T_{h^*}$ is reset at transition $tr_{h^*}$.
		Moreover, since $h^*$ is maximal and $(x, \lesssim, k) \in R_i$,
		by construction $(x, \lesssim, k) \in R_h$ for every $h^* \leq h \leq i$.
		Thus, $\hat x_{\lesssim k} \not\in T_h$ for every $h^* < h \leq i$.
		Therefore, $\mu_i(\hat x_{\lesssim k}) = t_{h^*, i}$.
		Since at the matching pop transition $tr_{j_{h^*}}$
		we have $z - x \lesssim k \in \psi_{j_{h^*}}$,
		and $x \in T_i$ is reset now,
		by Point 1) above we have $t_{h^*, i} \lesssim k$.
		Consequently, $\mu_i \models \hat x_{\lesssim k} \lesssim k$.
		
		\item $\mu_i \models \phi^1$:
		Immediately by the choice of $O'$.

	\end{itemize}
	Thus, $tr_i'$ is a valid transition.
	
	For the other inclusion, let $w = (a_0, t_0) \cdots (a_k, t_k)$ be a timed word accepted by $\uaut$,
	and let $\pi'$ be an accepting run:
	\begin{align*}
		\pi' = (r_0, \mu_0, u_0) tr_0' (r_1, \mu_1, u_1) \cdots (r_{k+1}, \mu_{k+1}, u_{k+1}),& \\
		\textrm{ with } r_i = (p_i, R_i, O_i)&
	\end{align*}
	We obtain an accepting run $\pi$ in $\taut$ by removing the extra components in the control state and stack alphabet,
	and by adding back pop constraints (as given by the symbol popped).
	\ignore{
	To show that $w$ is accepted by $\taut$,
	we introduce an intermediate model $\uaut'$,
	first prove that $w$ is accepted by $\uaut'$,
	and then, by a simple operation, show that $w$ is even accepted by $\taut$.
	The intermediate model $\uaut'$ is defined exactly as $\uaut$,
	but timing constraints on pop transitions are not forgotten.
	I.e., $\uaut'$ has pop transitions of the form
	\begin{align*}
		(p, R, O) \goesto {a, \phi', T, \pop{\alpha' \models \psi}} (q, R', O') & \\
		\textrm{ with } \alpha' = \tuple{\alpha, \psi, R', O'} &
	\end{align*}
	Since $w$ is accepted by $\uaut$,
	there exists an accepting run $\pi'$ over $w$ of the form
	\begin{align*}
		\pi' = (r_0, \mu_0, u_0) tr_0' (r_1, \mu_1, u_1) \cdots (r_{k+1}, \mu_{k+1}, u_{k+1}),& \\
		\textrm{ with } r_i = (p_i, R_i, O_i)&
	\end{align*}
	where $\mu_i$ is a clock valuation over $X'$ and $s_i$ is a timed word over $\Gamma'$,
	and $tr_i'$ is either a timed transition $tr_i' \in \R$,
	or a discrete transition of the form $tr_i' = (a_i, \phi_i, T_i, op_i)$.
	We show that $\pi'$ can be step-wise transformed into a run $\pi''$ in $\uaut'$.
	The only difficulty is to show that pop operations can be performed in $\uaut'$ by respecting their timing constraint.
	Formally,
	\begin{align*}
		\pi'' = (r_0, \mu_0, u_0) tr_0'' (r_1, \mu_1, u_1) \cdots (r_{k+1}, \mu_{k+1}, u_{k+1}),
	\end{align*}
	where
	$tr_i'' = tr_i'$ if $tr_i' \in \R$ or $tr_i' = (a_i, \phi_i, T_i, op_i)$ with $op_i$ either a nop or a push,
	and
	\begin{align}
		\label{eq:pop}
		tr_i'' = (a_i, \phi_i, T_i, \pop{\alpha_i' \models \psi_i}) & \\
		\nonumber
		\textrm{ with } \alpha_i' = \tuple{\alpha_i, \psi_i, R'_i, O'_i} &
	\end{align}
	if $tr_i' = (a_i, \phi_i, T_i, \pop{\alpha_i'})$.
	Notice that the only difference between $\pi'$ and $\pi''$
	is that pop transitions in $\pi''$ have the additional timed constraint $\psi_i$ (as recorded in the stack symbol).
	}
	To show that $\pi$ is an accepting run,
	we argue that $\mu_i \models \psi_i$ holds for a pop transition
	\begin{align*}
		tr_i' = (a_i, \phi_i, T_i, \pop{\alpha_i'}) & \\
		\nonumber
		\textrm{ with } \alpha_i' = \tuple{\alpha_i, \psi_i, R'_i, O'_i} &
	\end{align*}
	Let $tr_j'$ with $j < i$ be the corresponding push transition, i.e.,
	\begin{align}
		tr_j' = (a_j, \phi_j, T_j, \push{\alpha'_j}) & \\
		\nonumber
		\textrm{ with } \alpha_j' = \alpha'_i
	\end{align}
	Notice that the symbol popped at time $i$ matches the one pushed at time $j$.
	%
	We begin with reset restrictions.
	Let $z - x \lesssim k \in \psi_i$ any reset restriction on clock $x$ with $\lesssim \in \{ \leq, < \}$.
	We argue that $\mu_i \models z - x \lesssim k$ holds.
	The claim follows from the following observations. Let $j < h \leq i$ in the following:
	\begin{enumerate}

		\item Except possibly for the first push transition $j$,
		clock $\hat x_{\lesssim k}$ is never reset before, and including, transition $i$.
		This is follows from the fact that, once a reset obligation is added,
		it always subsumes new ones.
		Thus, $\mu_h(\hat x_{\lesssim k})$ is at least the age of $\alpha'_i$ at index $h$.
		
		\item $(x, \lesssim, k) \in R_h$.
		Indeed, $(x, \lesssim, k) \in R_{j+1}$ by construction.
		Moreover, nop operations do not change $R_h = R_{h+1}$,
		push operations put $R_h$ on the stack and $R_h \subseteq R_{h+1}$,
		and pop operations restore the $R_h$ of the corresponding push.
		
		\item By the previous point, if $x$ is reset (necessarily at a nop operation by assumption),
		then $\mu_h \models \hat x_{\lesssim k} \lesssim k$ holds by the definition of nop operation.
		
		\item Finally, $\mu_h \models -x \lesssim k$, for $h = j$.
		If $k \geq 0$, this is trivial.
		Otherwise, let $k < 0$,
		and let $j^*$ be the largest index $j^* \leq j$
		s.t. $(x, \lesssim, k) \in R_{j^*} \setminus R_{j^* - 1}$ is last added to reset obligations.
		Then, $\mu_{j^*} \models -x \lesssim k$ holds by definition of push operation.
		Since $k < 0$, $x$ is not reset ever since (cf. Point 3),
		and thus $\mu_j \models -x \lesssim k$.		
		
	\end{enumerate}
	There are two cases. If $x$ is reset between transition $j$ and $i$,
	then by 1) and 3) the age of $\alpha'_i$ is $\lesssim k$ the last time $x$ was reset.
	Consequently, at transition $i$, $\mu_i \models z - x \lesssim k$.
	If $x$ is not reset at all,
	then $\mu_j \models -x \lesssim k$ (by Point 4)
	immediately implies $\mu_i \models z - x \lesssim k$.
	
	We now consider reset obligations.
	Let $z - x \gtrsim k \in \psi_i$ with $\gtrsim \in \{ \geq, > \}$.
	We argue that $\mu_i \models z - x \gtrsim k$ holds.
	There are two cases. If $(x, \gtrsim, k) \not\in O_{j+1}$,
	then $\mu_j \models -x \gtrsim k$ holds by construction,
	i.e., the constraint must have been satisfied by a previous reset of $x$,
	which directly implies $\mu_i \models z - x \gtrsim k$.
	Now let $(x, \gtrsim, k) \in O_{j+1}$.
	We make the following observation.
	\begin{enumerate}
		
		
		\item[5)] $\hat x_{\gtrsim k}$ is at most the age of $\alpha_i'$.
		This is obvious, since $\hat x_{\gtrsim k}$ is reset at transition $j$ by construction.

	\end{enumerate}
	Since the pop at transition $i$ satisfies $O_i = \emptyset$,
	constraint $(x, \gtrsim, k)$ must be eventually removed.
	The only way to remove $(x, \gtrsim, k)$ from $O_h$
	is to either push it on the stack (cf. $O_0'$),
	or to reset $x$ when $\mu_h \models \hat x_{\gtrsim k} \gtrsim k$ holds (by definition of nop operation).
	In the former case, $(x, \gtrsim, k)$ will reappear in $O_h$ at the matching pop operation and still be pending.
	In the latter case, the age of $\alpha_i'$ was at least $k$ when $x$ was reset by Point 5) above,
	which directly implies $\mu_i \models z - x \gtrsim k$.
	\ignore{
	We conclude by showing that the run $\pi''$ in $\uaut'$ can be step-wise transformed into a run $\pi$ in $\taut$.
	Intuitively, we just remove the extra information added by the encoding,
	and we use the fact that 1) $\uaut'$ only adds new constraints and resets to $\taut$,
	2) the latter only involve the new clocks $\hat x$'s,
	and 3) $\pi''$ already contains the right pop timing constraints.
	Formally, we define $\pi$ as follows:
	\begin{align*}
		\pi = (p_0, \nu_0, v_0) tr_0 (p_1, \nu_1, v_1) \cdots (p_{k+1}, \nu_{k+1}, v_{k+1})
	\end{align*}
	where, for every $0 \leq i \leq k+1$,
	\begin{itemize}
		
		\item Clock valuation $\nu_i$ equals $\mu_i$ restricted to $X \subseteq X'$.
		\item Stack content $v_i$ is obtained from $u_i$ by projecting the stack alphabet $\Gamma'$ to $\Gamma$.
		Formally, if
		$u_i = (\tuple{\alpha_0, \psi_0, f_0, \tuple{Y_0, g_0}}, x_0) \cdots (\tuple{\alpha_n, \psi_n, f_n, \tuple{Y_n, g_n}}, x_n)$,
		then $v_i = (\alpha_0, x_0) \cdots (\alpha_n, x_n)$.
		
		\item Timed transitions do not change,
		i.e., $tr_i = tr_i''$ if $tr_i'' \in \R$.
		
		\item A discrete transition $tr_i$ is obtained from $tr_i''$ by projecting away the extra information encoded by the construction.
		Formally, if $tr_i'' = (a, \phi', T', op')$,
		then $tr_i = (a, \phi, T, op)$,
		where $\phi$ is obtained from $\phi'$ by removing the extra constraints involving clocks $\hat x$'s,
		$T = T' \cap X$,
		and, $op = \push\alpha$ if $op' = \push{\alpha, \psi, f, \tuple{Y, g}}$,
		$op = \pop{\alpha \models \psi}$ if $op' = \pop{\alpha, \psi, f, \tuple{Y, g} \models \psi}$,
		and $op = \nop$ if $op' = \nop$.
		
		Notice that in $\pi''$, a pop operation already has a timing constraint $\psi$.

	\end{itemize}
	It is immediate to see that $\pi$ thus constructed is a run in $\taut$.
	}
\end{proof}

\ignore{ 

\subsection{Previous simpler construction for non-diagonal constraints}

For a \dtPDA $\taut$ as above, we construct another \dtPDA
$\uaut = \tuple {Q', q_0', \Sigma, \Gamma', X', \Delta'}$ with timeless stack,
that is, with trivial age constraints $\alpha \in (-\infty, +\infty)$.
The idea is that $\uaut$ guesses in advance what will be the interval
in which the age of a newly pushed symbol will fall at the time of pop.
By using extra clocks for each guessed interval, we can make sure that the guess is indeed correct.
Moreover, since stack symbols age monotonically,
subsequently guessed intervals have endpoints that become smaller and smaller,
which implies that a finite number of extra clocks suffice.

For each interval $I$ equal to either $[n, n]$ or $(n-1, n)$
we have two clocks, $x_n$ and $y_n$.
The intuition is that $x_n$ measures the age of the \emph{first} symbol guessed to be in $I$,
and $y_n$ measures the age of the \emph{last} symbol guessed to be in $I$.
Thus, $x_n$ is reset to $0$ when the first symbol with age $I$ is pushed on the stack,
and $y_n$ is reset to $0$ at the last such time.
Then, at the time of the corresponding pop, 
if $I = [n, n]$, then we check that $n = x_n \wedge y_n = n$,
and if $I = (n-1, n)$, then we check that $n-1 < y_n$ and $x_n < n$
(notice that $y_n \leq x_n$ by construction at the time of pop).
Note that $x_n$ and $y_n$ are reset to $0$ at the same time if there is only one symbol guessed to have age in $I$.
Intervals $(-\infty, n)$ and $(n, +\infty)$ can be treated similarly.

Formally, control states of $\uaut$ are of the form $(p, I, f)$,
with $p \in Q$ a control state in $\taut$,
$I \in \mathcal I$ an interval for the guessed age of the next symbol that will be pushed on the stack,
and $f \in \{0, 1\}$ a flag that is used to keep track of the first and last occurrence of new intervals.
More precisely, $f = 1$ means that this is the first occurrence of an interval,
and $f = 0$ means that this is not the first occurrence of an interval.

Transitions in $\uaut$ are as follows.
If $p \goesto {a, \phi, Y, \pop{\alpha \in J}} q$ is a pop transition in $\taut$,
then we have the following pop transition in $\uaut$
\[ (p, I, f) \goesto {a, \phi', Y, \pop{(\alpha, J) \in (-\infty, +\infty)}} (q, I, f) \]
s.t., if $J = [n, n]$, then $\phi' = \phi \wedge x_n = n \wedge y_n = n$,
and, if $J = (n-1, n)$, then $\phi' = \phi \wedge x_n < n \wedge n - 1 < y_n $.

Moreover, if $p \goesto {a, \phi, Y, \push\alpha} q$ is a push transition in $\taut$,
then we have the following push transition in $\uaut$
\[ (p, I, f) \goesto {a, \phi, Y', \push{(\alpha, I)}} (q, J, g) \]
s.t. $I$ is either $[n, n]$ or $(n-1, n)$,
and $J \in \mathcal J$ is any guessed interval s.t. all the following conditions hold:
\begin{itemize}
	
	\item If $I = J$, then $g = 0$.
	\item If $I \neq J$, then $g = 1$.
	
	\item $Y'$ is the unique $Y \subseteq Y' \subseteq Y \cup \{x_n, y_n\}$ s.t.
	\begin{itemize}
		\item $x_n \in Y'$ if, and only if, $f = 1$.
		\item $y_n \in Y'$ if, and only if, $g = 1$.
	\end{itemize}

\end{itemize}

We now prove that $\lang {\taut} = \lang {\uaut}$.
First note that $\uaut$ only adds constraints and resets on the new clocks $x_n, y_n$.
For the left-to-right inclusion, let $w \in \lang {\taut}$ be a timed word accepted by $\taut$.
By appropriately guessing (at the time of push) the correct intervals in which the ages of stack symbols will fall at the time of pop,
and resetting the $x_n$'s and $y_n$'s accordingly,
it is easy to show that there exists a run for $w$ in $\uaut$.
For the other inclusion, let $w = (a_0, t_0) \dots (a_k, t_k) \in \lang {\uaut}$ be a timed word accepted by $\uaut$,
and let $\pi = (p_0, \mu_0, r_0) tr_0 (p_1, \mu_1, r_1) \cdots (p_{k+1}, \mu_{k+1}, r_{k+1})$
be a corresponding accepting run.
Consider the following forgetful projection function $pr$ from $\uaut$ to $\taut$,
defined on control states, clock valuations, stack contents,
and then lifted to configurations and finally rules:
\begin{align*}
	pr(p, I, f) &= p				\\
	pr((\alpha, I)) &= \alpha 	\\
	pr(\beta_0 \cdots \beta_h) &= pr(\beta_0) \cdots pr(\beta_h) \\
	pr(\mu) &= \lambda x \in X. \mu(x) \quad \textrm {(notice that $\mu$ is defined on $X' \supseteq X$)} \\
	pr(q, \mu, r) &= (pr(q), pr(\mu), pr(r)) \\
	pr(\phi) &= \textrm{ $\phi$ with occurrences of $x_n$ and $y_n$ removed } \\
	pr(Y) &= Y \cap X \\
	pr(\push\beta) &= \push{pr(\beta)} \\
	pr(\pop{(\alpha, I) \in J} &= \pop{\alpha \in I} \\
	pr(q \goesto {a, \phi, Y, op} q') &= pr(q) \goesto {a, pr(\phi), pr(Y), pr(op)} pr(q')
\end{align*}

}

\ignore{
\paragraph{Periodicity constraints make \dtPDA more powerful}
Consider the language $L$ of monotone timed palindromes s.t. matching symbols occur at integer distance,
i.e.,
\begin{align*}
	L = \bigcup_{n \geq 0} \{ &(a_0, t_0), \dots, (a_{2n-1}, t_{2n-1}) \ |\ \\
	&\forall (0 \leq i < n) \cdot a_{2n-i} = a_i \textrm{ and } t_{2n-i} - t_i \in \N \}
\end{align*}
Intuitively, $L$ cannot be recognized with a timeless stack (even if epsilon transitions are allowed),
since it contains precise dependencies between symbols that are arbitrarily far from each other (epsilon transitions could deal with these),
and, moreover, unboundedly many such dependencies.

However $L$ can be recognized with a timed stack and periodicity constraints.
First push $a_0, \dots, a_{n-1}$ (and initial age 0 for each symbol),
and then upon pop check that the clock on top of the stack is an integer, $z \equiv 0 \mod 1$.
}

 
\section{Proofs missing in Sec.~\ref{sec:defin}}

\BoundedSpanLemma*

\begin{proof}
	Every orbit, being defined by a minimal constraint, has uniformly bounded span. Therefore every orbit-finite set also does.
	
	For the opposite direction, if the span of the elements of an equivariant set $X\subseteq \Q^n$ is bounded
	by $k$, then $X$ is a subset of $$\{(x_1,\ldots, x_n) \ : \ \bigwedge_{i\neq j} x_i - x_j \leq k\}.$$
	The latter set can be equivalently defined by a finite disjunction of minimal constraints, and hence it is orbit-finite, 
	which implies orbit-finiteness of $X$.
\end{proof}

\NormalFormLemma*

\begin{proof}
	Fix a definable set $X$. Let $L$ be its indexing set. For each index $l \in L$ separately, 
	we compute a decomposition of $X_l$.

	Fix the index $l$. Let $K$ be larger than the largest absolute value of any integer constant used in the defining
	constraint of $X_l$, and let $d$ be the dimension of $X_l$.
	Enumerate all minimal constraints $\phi$ that define an orbit of span bounded by $(d-1)\cdot K$.
	Note that such constraints do not need to use integer constants of absolute value greater than $(d-1) \cdot K$.

	\begin{claim}
	It is decidable in polynomial time whether $\defin{\phi} \subseteq X_l$.
	\end{claim}
	\begin{proof}
	Indeed, for every pair of variables $x, y$ the minimal constraint $\phi$ determines an interval
	$I_{x, y}$ of the form
	\[
	\{ z \} \qquad \text{ or } \qquad (z, z+1),
	\]
	for $z \in \Z$, of possible values of $x - y$.
	In order to determine whether $\defin{\phi} \subseteq X_l$, we
	evaluate the constraint $\psi$ defining $X_l$ over the minimal constraint $\phi$, very much like a boolean formula
	is evaluated over a valuation of its variables.
	Atomic sub-formulae of $\psi$ are evaluated on the basis of the intervals $I_{x, y}$; for instance
	\[
	x - y < z
	\]
	evaluates to true if and only if $I_{x y} \subseteq \Z_{< z}$.
	$\defin{\phi} \subseteq X_l$ iff the constraint $\psi$ evaluates to true over $\phi$.
	\end{proof}
	Thanks to the claim, we compute  all minimal constraints satisfying $\defin{\phi} \subseteq X$
	and add them to the decomposition of $X$.

	The next claim formulates the weakness of constraints that we build upon:
	\begin{claim}
		let $O \subseteq \Q^d$ be an orbit.
		If $O \subseteq X$ and $O$ admits a gap $K$ then the $K$-extension of $O$ is included in $X$.
	\end{claim}
	Therefore, for every minimal constraint $\phi$ satisfying the above claim, we add to the decomposition of
	$X$ its $K$-extension.
	The decomposition is computed in \exptime, and its correctness follows by following last claim:
	\begin{claim}
		Every orbit $O \subseteq X$ of span larger than $(d-1)\cdot K$ is included in the $K$-extension of some orbit 
		$O' \subseteq X$ of span at most $(d-1)\cdot K$.
	\end{claim}
	This completes the proof of Lemma~\ref{lem:nf}.
\end{proof}


\section{Proofs missing in Sec.~\ref{sec:trPDA}}


\ThmUndecid*
\begin{proof}
The proof works for transition rules satisfying $n = 1$, $m = 2$ in~\eqref{eq:rules m n}.
The idea is to simulate a 2-counter Minsky machine $M$ by a \trPDA $\aaut_M$: 
one counter is simulated using the stack and two stack symbols $\bot, \top$, 
while the other counter is simulated using the difference between the time values stored in the top-most
stack symbol and in the state. 
It is enough if state space and stack alphabet are 1-dimensional, i.e., store exactly one time
value.
A configuration of $M$ with a control state $p$ and values of counters $n_1, n_2$ is represented
by the following configuration of $\aaut_M$:
\[
( \ (p, t+n_1+n_2), \ \ (\top, t+n_1) \ldots (\top, t+1) (\bot, t) \ )
\]
for an arbitrary $t \in \Q$ chosen nondeterministically by $\aaut_M$ in the beginning of the simulation.
The simulation assumes that the time values stored in consecutive stack symbols increase by 1, thus a push operation needs to see the current top-most stack symbol. 
Then increment (resp.~decrement) of the first counter is simulated by a simultaneous push (resp.~pop), and
increment (resp.~decrement) of the state by 1, e.g.: if $M$ increments the first counter and changes
state from $p$ to $p'$, the PDA has the following transition rule ($\mathtt{inc}_1$ is an input letter):
\[
( (p, t), (\top, u), \mathtt{inc}_1, (p', t+1), (\top, u+1) (\top, u) ).
\]
Operations on the second counter are performed exclusively on the time value stored in the state. 
Zero test $n_1 = 0$ of the first counter is done by checking if the top-most stack symbol is $(\bot, t)$ for an arbitrary $t \in \Q$; 
while zero test $n_2 = 0$ of the second counter is done by an equality test $t = u$ of the time values stored in the state and in 
the top-most symbol.

$\aaut_M$ accepts if $M$ halts from the control state $p$.
Thus the language $L(\aaut_M)$ is non-empty iff $M$ halts.
\end{proof}

\LemtrCFL*
\begin{proof}
	Let $\g$ be a \trCFG  with transition relation $\rho$, recognizing a timed language $L$.
	We show the untiming of $L$ can be recognized by a  CFG $\g'$ of size exponential in $\g$.
	Enumerate all orbits $O$ of $S$; this can be done effectively by Corollary~\ref{cor:nf}.
	$\g'$ will have a non-terminal $X_O$ for every orbit $O$ of $S$.
	For every non-terminals $X_O, X_{O_1}, \dots, X_{O_n}$ and for every orbit $P$ of $A_\varepsilon$,
	a production
	\begin{align*}
		(X_O,  P, X_{O_1}, \cdots, X_{O_n})
	\end{align*}
	is included in $\g'$ 
	whenever $\exists x \in O, a \in P, x_1 \in O_1, \dots, x_n \in O_n \cdot \rho(x, a, x_1, \dots, x_n)$ holds.
	The latter condition can be checked in EXPTIME, similarly like in the proof of Lemma~\ref{lem:nf}.
	Then $\g$ recognizes a timed word if, and only if, $\g'$ recognizes its untiming.
\end{proof}

\ThmCFG*
\begin{proof}
	The EXPTIME upper-bound follows immediately from Lemma~\ref{lem:trCFL}:
	From a \trCFG $\g$ recognizing a timed language $L$,
	we derive an exponentially larger context-free grammar $\g'$
	recognizing the untiming of $L$, for which non-emptiness is decidable in PTIME. 
	Correctness follows since $L$ is non-empty if, and only if, its untiming is non-empty.
	
	For the lower-bound, we reduce from the non-emptiness problem of the intersection
	of the languages recognized by $n$ (untimed) NFAs $\aaut_1, \dots, \aaut_n$ and a (untimed) CFG $\g$.
	(A similar reduction from this same problem was used in \cite{AbdullaAtigStenman:DensePDA:12} to show EXPTIME-hardness of \dtPDA).
	It is folklore that the latter problem is EXPTIME-hard;
	this can be shown by a direct reduction from linearly bounded alternating Turing machines.
	
	We adapt the textbook construction for intersection of a regular language 
	and a context-free one~\cite{HopcroftMotwaniUllman} in order to define a \tr context-free grammar $\g'$.
	We use timed registers to succinctly represent control states of the NFAs $\aaut_i$'s.
	%
	%
	Let $P_i$ be the set of control states of $\aaut_i$.
	For simplicity, we assume that $P_i = \set{1, \dots, k}$,
	that $1$ is the unique initial state of each NFA,
	and that $2$ is the unique final state of each NFA.
	A tuple of states of NFAs may be encoded as $ r \in \set{1, \dots, k}^{n}$.
	We write $ r \stackrel{a}{\longrightarrow}  r'$ if for every $i$, the pair of states at coordinate $i$
	in $ r$ and $ r'$,
	is related in the automaton $\aaut_i$ by an $a$-labelled transition.
	We will represent a pair of such tuples $( r,  r') \in \set{1, \dots, k}^{2n}$
	as an orbit in $\Q^{2n+1}$;
	one additional component will serve as reference point, and the others will be interpreted as the difference 
	w.r.t. the reference point.
	Thus, we encode $( r,  r')$ as the following orbit $O_{ r,  r'}$ in $\Q^{2n+1}$:
	\begin{align*}
		O_{ r,  r'} \ = \ \bigcup_{t \in \Q}(t,  r + t,  r' + t).
	\end{align*}
	Let symbols $S'$ of $\g'$ be 
	\[
	S' \ = \ S \times O_{ r,  r'}
	\]
	for $S$ the symbols of $\g$. Thus symbols in $S'$ are of the form
	\[
	(X, t,  r+t,  r'+t).
	\]
	Notice that $S'$ is orbit-finite.
	From the initial symbols, $\g'$ goes to one of the symbols 
	\[
	(X_0, t,  \vec 1 + t,  \vec 2 + t), \qquad \text{ for } t \in \Q,
	\]
	where $X_0$ is the initial symbol of $\g$, and $\vec 1$, $\vec 2$ are constant tuples.
	Assume for simplicity that $\g$ is in Chomsky normal form.
	For every production $X \goesto{} a$ in $\g$, the grammar $\g'$ has productions
	\[
	(X, t,  r+t,  r' + t) \goesto{} a
	\]
	for every $t \in \Q$, whenever $ r \stackrel{a}{\longrightarrow}  r'$.
	Moreover, for every production $X \goesto{} Y Z$ of $\g$, the grammar $\g'$ has productions
	\[
	(X, t,  r+t,  r'+t) \goesto{} (Y, t,  r+t,  r''+t) \ (Z, t,  r''+t,  r'+t), 
	\]
	for every $t \in \Q$ and
	for every three tuples $ r,  r',  r''$.
	
	The productions above are definable with (only equality) constraints of polynomial size.
	It is an easy exercise to check that the grammar $\g'$ recognizes the same language
	as the intersection of languages of $\aaut_1, \dots, \aaut_n$ and $\g$.
\end{proof}

\LemoftrPDA*
\begin{proof}
	Let $\Aa$ be an orbit-finite \trPDA.
	We define an orbit-finite \trPDA $\Bb$ in short form recognizing the same language.
	Intuitively, $\Bb$ keeps in the state a prefix of the stack long enough to apply rules of $\Aa$ without directly looking at the stack.
	Thus, states in $\Bb$ are pairs $\tuple{q, v}$ where $v \in S^*$ is a prefix 
	of a lhs/rhs of a rule of $\Aa$.
	Since projection and finite union preserve orbit-finiteness, $\Bb$ has an orbit-finite set of states.
	By Lemma~\ref{lem:proj} the set is definable and effectively computable.
	For every rule $(q, v, a, q', v')$ in $\Aa$ there exists a rule $\rhonop(\tuple{q, v}, a, \tuple{q', v'})$ in $\Bb$.
	Moreover, for every state $\tuple{q, v s}$ in $\Bb$,
	there exist rules $\rhopop(\tuple{q, v}, s, \varepsilon, \tuple{q, v s})$
	and $\rhopush(\tuple{q, v s}, \varepsilon, \tuple{q, v}, s)$
	in order to load/unload the local buffer of $\Bb$.
\ignore{ 
For a fixed an orbit-finite \trPDA $\aaut$, we define an \trPDA $\aaut'$ in short form.
The idea is to replace every transition rule $(q, v, a, q', v')$ by one $a$-labeled transition rule,
and a number of silent pop rules and silent push rules.
The state space of $\aaut'$ will consist additional auxiliary states; in particular,
lhs's and rhs's or transition rules of $\aaut$ will be states of $\aaut'$.

A special case when both $v$ and $v'$ are empty is simulated by a dummy push followed by a pop. 
Suppose that $v = s_1 \ldots s_n$ is nonempty (the remaining case when $v$ may be empty but 
$v'$ is nonempty is similar). 

The automaton $\aaut'$ executes a sequence of $n$ silent pops; the first one is
\[
\rhopop(q, s_1, \eps, (q, s_2, \ldots, s_n)) 
\]
and the following ones are of the form:
\[
\rhopop((q, s_i, \ldots, s_n), s_i, \eps, (q, s_{i+1}, \ldots, s_n)).
\]
As a next step $\aaut'$ 
and a similar sequence of silent push rules for consecutive symbols in $v'$.
A state $s_i$ stores  
The first of the transition rules should be labeled with $a$, if $a \neq \eps$.
}
The language is preserved by this transformation, and the size of $\Bb$ in
short form grows only polynomially with respect to the size of $\Aa$.
\end{proof}

\begin{lemma}
	Non-emptiness of \trPDA with timeless stack is \exptime-hard.
\end{lemma}

\begin{proof}
	As in \cite{AbdullaAtigStenman:DensePDA:12} (cf. Theorem~\ref{thm:trCFG}),
	we reduce from the non-emptiness problem of the intersection
	of the languages recognized by $n$ NFAs $\aaut_1, \dots, \aaut_n$ and a PDA $\baut$.
	This time, timed registers in the state are used to simulate the control states of the NFAs and the PDA,
	while the untimed pushdown simulates the pushdown of the PDA.
	We omit the details since they are very similar to Theorem~\ref{thm:trCFG}.
\end{proof}

 
\section{Proofs missing in Sec.~\ref{sec:proof}}

\subsection{Systems of equations}

\LemsetsofeqNP*

\begin{proof}
	NP-hardness of the membership problem follows from \cite{StockmeyerMeyer:1973},
	where it is shown that membership is already NP-hard when restricted to intersection-free systems with only non-negative constants $\set{0, 1}$.
	Moreover, the membership problem for $k \in \Z$ in $X$ easily reduces in polynomial time to non-emptiness  
	(by using intersection):
	it suffices to introduce a new variable $X'$ and a new inclusion $X' \supseteq X \cap \set k$.
	Then $X$ contains $k$ in the old system if, and only if, $X'$ is non-empty in the new system.
	The former inclusion can be simulated with only constants $\set{0, 1}$
	by looking at the binary representation of $k$
	and by introducing polynomially many new variables and inclusions.
	Thus, NP-hardness of the non-emptiness problem follows from NP-hardness of membership.

	It remains to show an NP upper bound for the non-emptiness problem.
	The procedure guesses in advance a sequence	of inclusions
	\[
		X_1 \ \supseteq \ Y_1 \cap \{0\}, \qquad \ldots, \qquad X_n \ \supseteq \ Y_n \cap \{0\}
	\]
	from $\Delta$, and then checks correctness of the guess by invoking membership tests.
	Let $\Delta'$ be obtained from $\Delta$ by removing all inclusions that use intersection.
	For every $0 \leq i < n$, let $\Delta_i$ be $\Delta'$ with the inclusions $X_1 \supseteq \{0\}, \dots, X_i \supseteq \{0\}$.
	The procedure checks that $0$ is in the least solution for $Y_{i+1}$ in $\Delta_i$.
	Each of these checks can be done in NP,
	as they reduce to non-emptiness of the intersection of the Parikh images of two context-free languages:
	\begin{itemize}
		\item the language of a context-free grammar over $\set{-1, 0, 1}$ obtained from $\Delta_i$
		by replacing addition with concatenation
		(as in the proof of Lemma~\ref{lem:setsofeq:PTIME}), and
		\item the language over $\set{-1, 0, 1}$ containing words with the same number of $-1$'s and $1$'s.
	\end{itemize}
	The non-emptiness of the intersection can be checked in NP by Kopczy{\'n}ski and To~\cite{KopczynskiTo:LICS2010}.
	
	It remains to argue for correctness.
	Let $\nu$ be the least solution for $\Delta$,
	and, for every $i$, let $\nu_i$ be the least solution for the guessed $\Delta_i$.
	By construction we have $\nu_1 \subseteq \nu_2 \subseteq \cdots \subseteq \nu_n \subseteq \nu$,
	therefore a right guess yields the correct answer.
	On the other side, suppose $k \in \nu(X)$ for some $k \in \Z$, and 
	let $t$ be a derivation of this fact constructed according to the following rules:
	\begin{align*}
		\frac{}{k \in X}
			&\textrm{ \quad for every } X \supseteq \set k \\[2ex]
					\frac{0 \in Y}{0 \in X}
			&\textrm{ \quad for every } X \supseteq Y \cap \set 0\\[2ex]
		\frac{k \in Y \quad l \in Z}{k + l \in X}
			&\textrm{ \quad for every } X \supseteq Y + Z
	\end{align*}
	The derivation is finite since we are considering least solutions.
	Given $t$, let $t_0, t_1, \dots, t_n$ be all sub-derivations (subtrees) of $t$
	s.t. $t_i$ proves a goal of the form $0 \in Y_{i+1}$.
	We further assume that $t_i$ is not a subtree of any previous $t_j$ with $1 \leq j < i$.
	The derivation $t_i$ can be used to show that $0$ belongs to $\nu_i(Y_{i+1})$.
	Thus, the algorithm correctly guesses and verifies $X_1 \supseteq Y_1 \cap \{0\}, \ldots, X_n \supseteq Y_n \cap \{0\}$.
\end{proof}
 
\subsection{Proof of Decomposition Lemma}

\DecompositionLemma*

\begin{proof}
The proof proceeds similarly as the proof of the Normal Form Lemma.
Consider the set of pairs of states extended with reference points (cf.~Sec.~\ref{sec:fromtrPDA:to:equations}):
\begin{align*}
	\ddot X = \setof{ (q\concat t, q'\concat t') \in {\dot Q}^2}{ (q, q') \in X, \ \ t, t' \in \Q }.
\end{align*}
The set $\ddot X$ is definable.
As in the proof of the Normal Form Lemma, we decompose $\ddot X$ into a finite union of orbits, 
and $K$-extensions of orbits $O$, for a sufficiently large positive
integer $K$, namely greater than the largest span of a state from $Q$.
Thus a state admits no gap $K$ or larger, and therefore a gap $K$ may only be caused by a large distance 
\emph{between} the reference points of two states.

Consider only those orbits $O \subseteq \ddot X$ where the difference of reference points is an integer (the property
is an invariant of an orbit); call these orbits \emph{integer-difference orbits}.

Every integer-difference orbit $O$ uniquely determines a pair
\begin{align} \label{eq:add1}
(o, z_0),
\end{align}
for $o$ an orbit in $\psQ$ and $z_0\in Z$ the difference of reference points, with $-2\cdot K \leq z_0 \leq 2\cdot K$.
Furthermore, consider $K$-extension of an integer-difference orbit $O$.
The integer-difference orbits included in the $K$-extension jointly determine one of the two sets,
\begin{align} \label{eq:add2}
\{o\} \times \Z_{\geq z_0} \qquad \text{ or } \qquad
\{o\} \times \Z_{\leq z_0}.
\end{align}
Therefore, the decomposition of the inverse image of $X$ contains singletons of all pairs listed in~\eqref{eq:add1},
and the sets listed in~\eqref{eq:add2}.
\end{proof}

\ignore{

\begin{proof}[Proof of Lemma~\ref{lem:decomp4}]
Lemma~\ref{lem:decomp4} is shown similarly as Lemma~\ref{lem:decomp}, with some
additional slightly technical and tedious details taken into account.

Fix a definable subset $X \subseteq Q^4$.
As before define $K$ and $k$.
Note that the span of elements of $\pspspsQ$ is bounded by $K+1$, as before.
As before we consider separately index quadruples indexing $Q^4$ and $\pspspsQ$.
For every fixed quadruple $(l, \bar l, \bar l', l')$, and for every fixed orbit $o$ in $\pspspsQ_{(l, \bar l, \bar l', l')}$,
we enumerate the image orbits $O$ of
\[
(o, z, \bar z, z'),
\]
for all $z, \bar z, z' \in \{ -M, \ldots M\}$ of sufficiently large absolute value.
The minimal constraint defining $O$ is computed from $(o, z, \bar z, z')$ similarly as before.
The safely large value of $M$ is larger than before, namely
\slcomm{right?}
\[
M \ = \ 3 (K + k + 1),
\]
as quadruples are now considered instead of tuples.
By the assumption that $X_{2 3}$ is orbit-finite we know by Span Lemma that the span of $X_{2 3}$ is bounded, and thus 
it is sufficient to consider the integer $\bar z$ of bounded absolute value; concretely, the sufficient bound is $K + k$.
Consider a fixed quadruple $(l, \bar l, \bar l', l')$ of indexes, a fixed orbit $o$  in $\pspspsQ_{(l, \bar l, \bar l', l')}$, 
and a fixed value 
\[
\bar z \in \{ -(K+k), \ldots, K+k \}
\]
from now on. 

Let $x, x'$ range below over all integers, contrarily to $z, z'$ than range over $\{ -M, \ldots, M\}$ only.
Call an item $(o, x, \bar z, x')$ \emph{contributing} if the image orbit $O$ of $(o, x, \bar z, x')$ 
is included in $X_{(l, \bar l, \bar l', l')}$.
Our aim is to compute a decomposition of the set of splits of all contributing items, i.e., of pairs
\begin{align} \label{eq:splits}
((o_{2 3}, \bar z), (o_{1 4}, x + \bar z + x')).
\end{align}
We separately treat 'small' items, namely those satisfying $x + \bar z + x' \in\{-M, \ldots, M\}$, 
and the remaining 'large' ones.

First, consider a small contributing item $(o, x, \bar z, x')$, i.e., $s = x + \bar z + x' \in \{-M, \ldots, M\}$.
If one of $x, x'$ is outside of $\{-M, \ldots, M\}$ then we may apply Gap Lemma to deduce that
for some small enough $z, z' \in \{-M, \ldots, M\}$, there is a contributing item $(o, z, \bar z, z')$ with the same
value of the sum, i.e., $s  = z + \bar z + z'$.
Therefore, in order to capture the splits~\eqref{eq:splits} of all small items, 
it is enough to add to the decomposition all  single-point sets
\begin{align} \label{eq:sets}
(\{o_{2 3}\} \times \{\bar z\}) \times (\{o_{1 4}\} \times \{z + \bar z + z'\}),
\end{align}
for all contributing small items $(o, z, \bar z, z')$.
As before, it is decidable in polynomial time if an item $(o, z, \bar z, z')$ is contributing.

In case of large items, we will replace the singleton
$\{z + \bar z + z'\}$ in~\eqref{eq:sets} with an infinite integer interval, as specified below.

We distinguish two symmetric cases. Consider a large contributing item $(o, x, \bar z, x')$, 
i.e., $x + \bar z + x' \geq M$.
As $M$ is large enough, one of $x, x'$ must be large enough as well, namely at least $K+k+1$.
Thus we may apply Gap Lemma to deduce that $(o, z, \bar z, z')$ is contributing as well,
for some $z, z' \in \{ -M, \ldots, M\}$ with $z + \bar z + z' = M$.
Therefore, we add to the decomposition the following set:
\begin{align*}
& (\{o_{2 3}\} \times \{\bar z\}) \times (\{o_{1 4}\} \times \Z_{> M}),
\end{align*}
for all contributing large items $(o, z, \bar z, z')$ with $z + \bar z + z' = M$.
Symmetrically, we also add to the decomposition the set:
\begin{align*}
& (\{o_{2 3}\} \times \{\bar z\}) \times (\{o_{1 4}\} \times \Z_{< -M}),
\end{align*}
for all contributing large items $(o, z, \bar z, z')$ with $z + \bar z + z' = -M$.
This completes the proof of Lemma~\ref{lem:decomp}.
\end{proof}

}

\end{document}